\documentclass{amsart}
\pdfoutput=1

\usepackage{amssymb}
\usepackage{amsthm}
\usepackage{amsfonts}
\usepackage{amsmath}
\usepackage{pmboxdraw}
\usepackage{verbatim} 
\usepackage{graphicx}
\usepackage{color}
\usepackage[colorlinks=true, citecolor=blue, filecolor=black, linkcolor=black, urlcolor=black]{hyperref}
\usepackage{cite}
\usepackage[normalem]{ulem}
\usepackage{subcaption}
\usepackage{bbm}
\usepackage{bm}
\usepackage{mathtools}
\usepackage{todonotes}
\usepackage{kantlipsum}
\allowdisplaybreaks
\usepackage{amsaddr}

\newcommand{\RN}[1]{%
	\textup{\uppercase\expandafter{\romannumeral#1}}%
}

\def\bfs{\boldsymbol}

\def\C{\mathbb{C}}

\def\P{\mathbf{P}}
\def\R{\mathbb{R}}

\newcommand{\Pf}{{\textup{Pf}}}

\newcommand{\re}{\operatorname{Re}}
\newcommand{\im}{\operatorname{Im}}

\theoremstyle{plain}
\newtheorem{thm}{Theorem}[section]

\newtheorem{cor}[thm]{Corollary}
\newtheorem{lem}[thm]{Lemma}

\newtheorem{prop}[thm]{Proposition}

\theoremstyle{remark}
\newtheorem{defn}{Definition}

\newtheorem{rem}{Remark}
\newtheorem{ex}{Example}

\newcommand{\abs}[1]{\lvert#1\rvert}
\numberwithin{equation}{section}

\begin{document}

\title[Spectral moments of complex and symplectic non-Hermitian random matrices]{Spectral moments of complex and symplectic \\ non-Hermitian random matrices}
\author{Gernot Akemann}
\address{Faculty of Physics, Bielefeld University, P.O. Box 100131, 33501 Bielefeld, Germany}
\email{akemann@physik.uni-bielefeld.de}

\author{Sung-Soo Byun}
\address{Department of Mathematical Sciences and Research Institute of Mathematics, Seoul National University, Seoul 151-747, Republic of Korea}
\email{sungsoobyun@snu.ac.kr}

\author{Seungjoon Oh}
\address{Department of Mathematical Sciences, Seoul National University, Seoul 151-747, Republic of Korea}
\email{seungjoonoh@snu.ac.kr}

\thanks{The authors are grateful to the DFG-NRF International Research Training Group IRTG 2235 supporting the Bielefeld-Seoul graduate exchange programme (NRF-2016K2A9A2A13003815). 
Sung-Soo Byun was supported by the National Research Foundation of Korea grant (RS-2023-00301976, RS-2025-00516909).}

\begin{abstract} 
We study non-Hermitian random matrices belonging to the symmetry classes of the complex and symplectic Ginibre ensemble, and present a unifying and systematic framework for analysing mixed spectral moments involving both holomorphic and anti-holomorphic parts. For weight functions that induce a recurrence relation of the associated planar orthogonal polynomials, we derive explicit formulas for the spectral moments in terms of their orthogonal norms. This includes exactly solvable models such as the elliptic Ginibre ensemble and non-Hermitian Wishart matrices. In particular, we show that the holomorphic spectral moments of complex non-Hermitian random matrices coincide with those of their Hermitian limit up to a multiplicative constant, determined by the non-Hermiticity parameter. Moreover, we show that the spectral moments of the symplectic non-Hermitian ensemble admit a decomposition into two parts: one corresponding to the complex ensemble and the other constituting an explicit correction term. This structure closely parallels that found in the Hermitian setting, which naturally arises as the Hermitian limit of our results. Within this general framework, we perform a large-$N$ asymptotic analysis of the spectral moments for the elliptic Ginibre and non-Hermitian Wishart ensemble, revealing the mixed moments of the elliptic and non-Hermitian Marchenko--Pastur laws. Furthermore, for the elliptic Ginibre ensemble, we employ a recently developed differential operator method for the associated correlation kernel, to derive an alternative explicit formula for the spectral moments and obtain their genus-type large-$N$ expansion.
\end{abstract}

\maketitle

\section{Introduction}
 
In his seminal 1955 paper~\cite{Wi55}, Wigner introduced the use of eigenvalue statistics of random matrices to model physical observables, particularly the energy level spacings in heavy atomic nuclei. In a subsequent work~\cite{Wi58}, he studied the \emph{spectral moments} of random matrices—fundamental quantities that encode the distribution of eigenvalues—and used it to derive the celebrated Wigner semicircle law for Gaussian random matrices.

Since then, spectral moments have been extensively investigated across a wide range of random matrix ensembles and have found deep applications in both mathematics and physics. Notably, in \cite{HZ86}, Harer and Zagier derived a recurrence relation for the spectral moments of the Gaussian Unitary Ensemble (GUE), which they used to compute the Euler characteristic of the moduli space of algebraic curves. The spectral moments of the Gaussian Orthogonal and Symplectic Ensemble (GOE and GSE) have also been extensively studied, see e.g. \cite{Le09,MS11}. In addition, general Gaussian $\beta$-ensembles have been investigated from the perspective of symmetric function theory, particularly via Jack polynomials~\cite{GJ97,Ok97}.

Beyond the classical ensembles with Gaussian weights, spectral moments have been analysed in orthogonal polynomial ensembles \cite{HT03,CDO21,Di03,CMSV16,GGR21,CMOS19,RF21} and their $q$-deformations—often referred to as quantised ensembles in the physics literature~\cite{FLSY23,CCO20,MPS20,Le05,BFO24}. These developments have led to significant applications in diverse areas, including the theory of quantum dots and $\tau$-functions~\cite{MS11,MS12,MS13,Cu15,CMSV16,LV11}.

The moments of Hermitian ensembles possess structural properties that go beyond explicit summation expressions. For the GUE, they satisfy a second-order recurrence relation~\cite{Le05,HT03,WF14}, and their exponential generating function admits a closed-form expression in terms of a hypergeometric function~\cite{Le05,HT03}. In the case of the Laguerre Unitary Ensemble (LUE), the moments also obey a second-order recurrence~\cite{Le05,HT03,RF21}, and, upon differentiation, their exponential generating function can similarly be expressed via a hypergeometric function~\cite{HT03,Fo21}. For the Jacobi Unitary Ensemble (JUE), the moments satisfy a third-order recurrence~\cite{Le05,CMOS19,RF21}. These additional structural features are inherited by the holomorphic spectral moments of the corresponding non-Hermitian ensembles.

In recent years, increasing attention has been paid to random matrices without Hermiticity constraint, where the eigenvalues are distributed in the complex plane. This area, known as non-Hermitian random matrix theory~\cite{BF25}, takes the Ginibre ensemble—matrices with independent and identically distributed Gaussian entries—as its prototypical model. Spectral moments of such non-Hermitian ensembles, especially those with real entries (known as the real Ginibre ensemble or Ginibre orthogonal ensemble), have been actively studied~\cite{EKS94,FN07,SK09,BF24,By24b,BN25,FR09}. In these cases, the analysis must separately treat real and complex eigenvalues. In particular, even the zeroth moment of the real eigenvalue spectral density is highly nontrivial, as it corresponds to the expected number of real eigenvalues~\cite{EKS94,FN07}.

In contrast, the spectral moments of non-Hermitian random matrices in the symmetry classes corresponding to Ginibre unitary and symplectic ensemble (GinUE and GinSE), remain less explored in the literature. In this work, we aim to contribute to this direction by proposing a systematic framework for their analysis, based on recently developed techniques in the theory of planar (skew)-orthogonal polynomials.
 
\medskip 

We begin by introducing our models.
For a given weight function $\omega:\mathbb{C} \to \mathbb{R}$, we consider points of configurations $\bfs{z}=\{z_j\}_{j=1}^N$ with joint probability density functions
\begin{align}
d\P_N^{ \mathbb{C} }(\boldsymbol{z}) & = \frac{1}{Z_N^{ \mathbb{C} } } \prod_{j<k} |z_j-z_k|^2 \prod_{j=1}^N \omega(z_j) \, \,dA(z_j), \label{Gibbs complex} 
\\
d\P_N^{ \mathbb{H} }(\boldsymbol{z}) & = \frac{1}{Z_N^{ \mathbb{H} } } \prod_{j<k} |z_j-z_k|^2 |z_j-\overline{z}_k|^2 \prod_{j=1}^N |z_j-\overline{z}_j|^2 \omega(z_j) \, \,dA(z_j), \label{Gibbs symplectic}
\end{align}
where $dA(z)=d^2z/\pi$ is the area measure. Here $Z_N^{ \mathbb{C} }$ and $Z_N^{ \mathbb{H} }$ are normalisation constants known as the partition functions.
A fundamental example arises when the weight is given by $\omega(z) = e^{-|z|^2}$, in which case the ensembles~\eqref{Gibbs complex} and~\eqref{Gibbs symplectic} correspond to the eigenvalue distributions of the GinUE and GinSE, respectively~\cite{BF25}. 
In general, the ensembles~\eqref{Gibbs complex} and~\eqref{Gibbs symplectic} correspond to different two-dimensional Coulomb gases with inverse temperature $\beta = 2$~\cite{Fo10,Se24}. Furthermore, these are also known as the random normal matrix ensemble and the planar symplectic ensemble, respectively. 

\begin{defn}
 For $p_1,p_2 \in \mathbb{Z}_{ \ge 0 }$, the spectral moments of the ensembles $\bfs z$ are defined by 
\begin{equation} \label{def of spectral moments}
M^{ \mathbb{C} }_{p_1,p_2,N} := \mathbb{E}^{ \mathbb{C} }_N\bigg[\sum_{j=1}^N z_j^{p_1} \overline{z}_j^{ p_2 } \bigg], \qquad M^{ \mathbb{H} }_{p_1,p_2,N} := \mathbb{E}^{ \mathbb{H} }_N\bigg[\sum_{j=1}^N z_j^{p_1} \overline{z}_j^{ p_2 } \bigg] ,   
\end{equation}
where the expectations $\mathbb{E}^{ \mathbb{C} }_N$ and $\mathbb{E}^{ \mathbb{H} }_N$ are taken with respect to the probability measures $\P_N^{ \mathbb{C} }$ and $\P_N^{ \mathbb{H} }$, respectively. 
\end{defn}

We note that, in contrast to the Hermitian case (see \eqref{hermitian moment} below), for general non-Hermitian random matrices $X$, the right-hand side of \eqref{def of spectral moments} cannot be expressed as the expectation of $\textup{Tr}(X^{p_1}(X^{\dagger})^{p_2})$.


From the perspective of non-Hermitian random matrix theory and orthogonal polynomial theory, weight functions belonging to the following classes are of particular interest. Each weight depends on a non-Hermiticity parameter $\tau \in [0,1)$, where the limit $\tau \uparrow 1$ corresponds to a degeneration of the weight function to one defined on the real line.
\begin{itemize}
    \item \textbf{Planar Hermite weight}. For $\tau \in [0,1)$, let 
    \begin{equation} \label{def of weight H}
    \omega^{ \rm H }(z):= \exp\Big( -\frac{ |z|^2-\tau \re z^2  }{1-\tau^2} \Big).  
    \end{equation}
    In the Hermitian limit, we have 
    \begin{equation} \label{def of weight H Hermitian}
    \lim_{\tau \uparrow 1} \omega^{ \rm H }(x+iy) = \omega_\R^{ \rm H }(x) \mathbbm{1}_{ \{ y=0 \} } , \qquad \omega_\R^{ \rm H }(x):=e^{-x^2/2}.  
     \end{equation}
    \item \textbf{Planar Laguerre weight}. For $\tau \in [0,1)$ and $\nu >-1$, let
\begin{align} \label{def of weight L}
    \omega^{ \rm L }(z):= |z|^\nu K_\nu\Big(\frac{ 2|z| }{1-\tau^2} \Big) \exp\Big(\frac{2\tau}{1-\tau^2}\re z\Big), 
\end{align}
where $K_\nu$ is the modified Bessel function of the second kind \cite[Chapter 10]{NIST}.
 In the Hermitian limit, we have 
    \begin{equation}  \label{def of weight L Hermitian}
    \lim_{\tau \uparrow 1} \omega^{ \rm L }(x+iy) =   \omega_\R^{ \rm L }(x) \cdot \mathbbm{1}_{ \{ y=0 \} }  , \qquad  \omega_\R^{ \rm L }(x):= x^\nu e^{-x} \cdot \mathbbm{1}_{ \{ x>0 \} }.  
     \end{equation}
    \item \textbf{Planar Gegenbauer weight}. For  $\tau \in [0,1)$ and $a>-1$, let
\begin{equation} \label{def of weight G}
   \omega^{ \rm G }(z) = \Big( 1 - \frac{2 ( \re z)^2}{1+\tau} - \frac{2 (\im z)^2}{1-\tau} \Big)^a \cdot \mathbbm{1}_{K}(z), \qquad   K: = \Big\{ (x,y)\in \R^2  :  \frac{2x^2}{1+\tau} + \frac{2y^2}{1-\tau} \le 1 \Big\}.
\end{equation} 
In the Hermitian limit, we have 
    \begin{equation}  \label{def of weight G Hermitian}
    \lim_{\tau \uparrow 1} \omega^{ \rm G }(x+iy) =   \omega_\R^{ \rm G }(x) \cdot    \mathbbm{1}_{ \{ y=0 \} }  , \qquad  \omega_\R^{ \rm G }(x) := (1-x^2)^a  \cdot \mathbbm{1}_{ \{ |x|<1 \} }  .
     \end{equation}
\end{itemize}
 
As will be discussed below, the ensembles \eqref{Gibbs complex} and \eqref{Gibbs symplectic} with the weight functions are exactly solvable in the sense that their correlation functions can be explicitly analysed by virtue of classical orthogonal polynomials. Moreover, the ensembles \eqref{Gibbs complex} and \eqref{Gibbs symplectic} with planar Hermite and Laguerre weights provide realisations of important non-Hermitian random matrix ensembles, namely the \textit{elliptic Ginibre ensemble} and the \textit{non-Hermitian Wishart ensemble}, respectively. See \cite[Sections 2.3 and 10.5]{BF25} and \cite{BN24} and references therein. The ensemble \eqref{Gibbs complex} with the planar Gegenbauer weight has been studied in~\cite{ANPV21,NAKP20,Na24}.
In addition, the weight functions $\omega_\R^{\rm H}$, $\omega_\R^{\rm L}$, and $\omega_\R^{\rm G}$ on the real axis are classical in the theory of orthogonal polynomials, corresponding respectively to the Hermite, Laguerre, and Gegenbauer (symmetric Jacobi) weights.

In what follows, we will use the superscripts $\mathrm{H}$, $\mathrm{L}$, and $\mathrm{G}$ to indicate quantities associated with the weight functions $\omega^{\rm H}$, $\omega^{\rm L}$, and $\omega^{\rm G}$, respectively. For example,
$M_{p_1,p_2}^{\rm H, \mathbb{C}}$, $M_{p_1,p_2}^{\rm L, \mathbb{H}},$ and so on.

Note also that the elliptic Ginibre ensembles reduce to the weight $e^{-|z|^2}$ when $\tau = 0$. In this case, the ensembles \eqref{Gibbs complex} and \eqref{Gibbs symplectic} associated with $\omega^{\rm H}|_{\tau=0}$ correspond to GinUE and GinSE respectively. To indicate these specific cases, we add superscripts to the associated quantities, such as $M_{p_1,p_2}^{\rm GinSE}$.

\bigskip 

We now discuss planar orthogonal polynomials, which serve as the primary tools for the analysis carried out in this work.
Let $d\mu(z) = \omega(z)\,dA(z)$ be a measure on $\mathbb{C}$ with real moments. Then, we have an inner product on the space of polynomials with real coefficients
\begin{equation} \label{inner product}
    \langle f,g \rangle := \int_\mathbb{C} f(z) \overline{g(z)} \, d\mu(z).
\end{equation}
For a given weight function $\omega$, let $(p_k)_{ k =1 }^\infty$ be a family of \textit{monic} orthogonal polynomial satisfying
\begin{equation} \label{OP norm}
 \langle p_j,p_k \rangle = h_k \,\delta_{j,k},
\end{equation}
where $h_k$ is the squared norm and $\delta_{j,k}$ is the Kronecker delta. 
 
The weight functions $\omega^{ \rm H } $, $\omega^{ \rm L }$, and $\omega^{ \rm G }$ defined in \eqref{def of weight H}, \eqref{def of weight L}, and \eqref{def of weight G}, respectively, admit closed-form expressions for the associated planar orthogonal polynomials in terms of classical orthogonal polynomials.

\begin{itemize}
    \item \textbf{Planar Hermite polynomials}. For the weight $\omega^{ \rm H }$, it was shown in \cite{EM90,DGIL94} that the associated planar orthogonal polynomials and squared norms are given by 
\begin{align} \label{def of planar Hermite}
p_k^{ \rm H }(z) := \Big(\frac{\tau}{2}\Big)^{k/2} H_k \Big( \frac{z}{\sqrt{2\tau}} \Big), \qquad  h_k^{\rm H} := k! \sqrt{1-\tau^2},
\end{align}
where $H_k$ is the Hermite polynomial \eqref{def of Hermite}.
\smallskip 
\item \textbf{Planar Laguerre polynomials}. For the weight $\omega^{ \rm L }$, it was shown in \cite{Ak05,Ka01,Os04} that the associated planar orthogonal polynomials and squared norms are given by 
\begin{align} \label{def of planar Laguerre}
p_k^{ \rm L }(z) := (-1)^k k! \tau^k L_k^\nu\Big(\frac{z}{\tau}\Big), \qquad h_k^{ \rm L } := \frac{1-\tau^2}{2} k!\,  \Gamma(k+\nu+1), 
\end{align}
where $L^\nu_k$ is the generalised Laguerre polynomial \eqref{def of Laguerre}. 
\smallskip 
\item \textbf{Planar Gegenbauer polynomials}. For the weight $\omega^{ \rm G }$, it was shown in \cite{ANPV21} that the associated planar orthogonal polynomials and squared norms are given by 
\begin{align} \label{def of planar Gegenbauer}
p_k^{ \rm G }(z) := \frac{k!}{(1+a)_k} \Big(\frac{\sqrt{\tau}}{2}\Big)^k C^{(1+a)}_k \Big(\frac{z}{\sqrt{\tau}}\Big), \qquad h_k^{ \rm G } := \sqrt{1-\tau^2} \frac{1+a}{k+1+a} \Big( \frac{k!}{(1+a)_k} \Big)^2 \Big( \frac{\tau}{4} \Big)^k C^{(1+a)}_k\Big(\frac{1}{\tau}\Big),
\end{align}
where $C_k^{(a)}$ is the Gegenbauer (symmetric Jacobi) polynomial \eqref{def of Gegenbauer}.  Here, $(a)_k=a (a+1)\dots (a+k-1)$ is the Pochhammer symbol. 
\end{itemize}

We now turn to the classical theory of orthogonal polynomials in one real variable. To this end, consider a linear functional $\mathcal{L} : \mathbb{R}[x] \to \mathbb{R}$ defined on the space of real polynomials: for a weight $\omega_\R: \R \to \R_+$, 
\begin{equation}
\mathcal{L}(P(x))= \int_\R P(x) \omega_\R(x)\,dx. 
\end{equation} 
Then, the monic orthogonal polynomials $P_j$ are characterised by 
\begin{equation} \label{def of real OP}
\mathcal{L}( P_j(x) P_k(x) ) = \mathcal{L}( P_j(x) P_j(x) ) \, \delta_{j,k}. 
\end{equation} 

From the random matrix viewpoint, the associated ensemble of interest has the joint eigenvalue probability density function proportional to 
\begin{equation} \label{def of Hermitian beta ensemble}
\prod_{ j <k } |\lambda_j-\lambda_k|^\beta \prod_{j=1}^N \omega_\R(\lambda_j)\,d\lambda_j, \qquad \lambda_j \in \mathbb{R}.
\end{equation}
Here, $\beta > 0$ denotes the Dyson index, taking the values $1$, $2$, and $4$, corresponding to the orthogonal (O), unitary (U), and symplectic (S) ensembles, respectively \cite{Fo10}. 
For Hermitian random matrices, there is no need to distinguish between holomorphic and anti-holomorphic moments, and the spectral moments are simply defined by
\begin{equation} \label{hermitian moment}
M_{p,N}^{ \beta }:= \mathbb{E}_N^{\beta}\Big[ \sum_{j=1}^N \lambda_j^p \Big] = \mathbb{E}_N^X\Big[\text{Tr}(X^p)\Big], 
\end{equation}
where the expectation $\mathbb{E}_N^\beta$ is taken with respect to the measure \eqref{def of Hermitian beta ensemble}, and $\mathbb{E}_N^X$ with the probability measure for the corresponding matrix ensemble $X$. 

For the weight functions $\omega_\R^{\rm H}$, $\omega_\R^{\rm L}$, and $\omega_\R^{\rm G}$ given in \eqref{def of weight H Hermitian}, \eqref{def of weight L Hermitian}, and \eqref{def of weight G Hermitian}, respectively, the ensemble \eqref{def of Hermitian beta ensemble} corresponds to the eigenvalue distributions of the Gaussian, Laguerre (L), and (symmetric) Jacobi (J) ensemble. These are commonly referred to as the GOE/GUE/GSE, LOE/LUE/LSE, and JOE/JUE/JSE, depending on the value of the Dyson index~$\beta$. As before, we use superscripts to indicate the quantities associated with specific ensembles--for example, $M_{p,N}^{\rm LSE}$.

\medskip 

Next, we recall the definitions of inversion and linearisation coefficients, which have been extensively studied from a combinatorial perspective on orthogonal polynomials (see e.g. \cite{CKS16} and references therein). We also refer the reader to \cite{Ans05} for derivations of linearisation coefficients using a stochastic process approach.

\begin{defn} \label{Def_mixed and linearisation} Let $(P_j)_{j=1}^{\infty}$ be a given set of orthogonal polynomials satisfying \eqref{def of real OP}. 
\begin{itemize}
    \item The \emph{inversion coefficient} $a_{n,k}$ is defined by
\begin{equation} \label{inverse representation}
x^n = \sum_{k=0}^n  a_{n,k}  P_k(x), \qquad  a_{n,k}:= \frac{ \mathcal{L}( x^n P_k(x) ) }{ \mathcal{L}( P_k(x)^2 ) }. 
\end{equation} 
    \item The \emph{linearisation coefficient} $b_{n,m,k}$ is defined by 
    \begin{equation} \label{linearisation}
P_n(x) P_m(x) = \sum_{k=0}^{n+m} b_{n,m,k} P_k(x), \qquad b_{n,m,k}:= \frac{ \mathcal{L}( P_n(x) P_m(x) P_k(x) ) }{ \mathcal{L}( P_k(x)^2 ) }.  
\end{equation}
\end{itemize} 
\end{defn} 

One significant advantage of these quantities is that they admit explicit formulas established in the literature for a broad class of orthogonal polynomials, typically within the Askey scheme. As a result, many observables of interest can be evaluated in closed form. In particular, explicit expressions for the inversion and linearisation coefficients associated with the Hermite, Laguerre, and Gegenbauer polynomials are presented in Subsection~\ref{Subsec_OP basics}.

\section{Main results} 

In this section, we present our main results. A brief overview is as follows:

\begin{itemize}
    \item In Theorem~\ref{thm: main}, we provide a general closed-form expression for the mixed spectral moments for a certain class of weight functions. This result yields explicit formulas for exactly solvable models, see, for example, Corollary~\ref{Cor_eGinibre moments} for the elliptic Ginibre ensembles.
    \smallskip
    \item In Theorem~\ref{Thm_elliptic and MP}, we focus on the planar Hermite and Laguerre weights, for which the associated limiting spectral distributions are known in the literature. For these cases, we derive the leading-order asymptotics of the spectral moments as $N \to \infty$.
    \smallskip
    \item In Theorem~\ref{Thm_eGinibre 2}, we derive an alternative expression for the spectral moments of the elliptic Ginibre ensemble using a recently developed method \cite{LR16,BE23,BES23}. This leads to a genus-type large-$N$ expansion, as demonstrated in Theorem~\ref{Thm_eGinUSE genus expansion}.
\end{itemize}

\subsection{Spectral moments of non-Hermitian random matrix ensembles}

Throughout this work, we focus on the class of planar orthogonal polynomials $p_k$ that satisfy the three term recurrence relation:  
\begin{equation} \label{eq: recurrence}
    z p_k(z) = p_{k+1}(z) + b_k p_k(z) + c_k p_{k-1}(z).
\end{equation}
We remark that, in general, planar orthogonal polynomials do not satisfy a recurrence relation, in contrast to their real-variable counterparts. Nevertheless, several notable examples do exhibit such a structure. In fact, all known exactly solvable models in non-Hermitian random matrix theory--particularly variants of the Ginibre ensembles \cite{BF25}--possess this property, as it governs the structure described in~\eqref{pk Pk scaling} below. Another important reason for imposing such a condition is that it enables the construction of skew-orthogonal polynomials, which serve as fundamental building blocks for analysing planar symplectic ensembles under this assumption~\cite{AEP22}. (We also refer the reader to \cite{ABN24} for an alternative and more general framework for constructing skew-orthogonal polynomials.)

For $p\in\mathbb{Z}_{\ge0}$, consider the linear map
\begin{equation}
    T_p f(z) := z^p f(z).
\end{equation}
By using \eqref{eq: recurrence}, we define $ (A^p)^j_k$ as the coefficients of the expansion
\begin{equation} \label{eq: Tp}
    T_p \, p_k(z) = \sum_{j=k-p}^{k+p} (A^p)^j_k \, p_j(z), 
\end{equation}
where $(A^p)^j_k=0$ for $j<0$. Note also the trivial case $p=0$:
\begin{equation*} \label{Ap trivial}
    (A^p)^j_k = \delta_{j,k}.
\end{equation*}

\begin{lem} \label{lem: scaling} Suppose that 
\begin{equation} \label{pk Pk scaling}
p_k(z) = \alpha^k P_k\Big( \frac{z}{ \alpha }\Big), 
\end{equation}
where $P_k$ is a set of orthogonal polynomials with respect to a weighted Lebesgue measure on $\mathbb{R}$. 
Then, 
\begin{equation}
(A^p)_{k}^j = \alpha^{p+k-j} \sum_{ l= j-k 	\vee 0 }^p  a_{p,l} \, b_{l,k,j} ,  
\end{equation}
where $a_{n,k}$ and $b_{n,m,k}$ are the inversion and linearisation coefficients of $P_n$ given in \eqref{inverse representation} and \eqref{linearisation}. 
\end{lem}

This lemma will be shown in Subsection~\ref{Subsec_OP basics}. 
Explicit formulas for the coefficients $(A^p)_k^j$ corresponding to planar Hermite, Laguerre, and Gegenbauer polynomials are presented in Proposition~\ref{Prop_Ap jk explicit}. We note that the scaling factor $\alpha$ equals $\sqrt{\tau}$ for the planar Hermite and Gegenbauer polynomials, and $\tau$ for the planar Laguerre polynomials. Also, notice that
\begin{equation} \label{Ap proportional}
(A^p)_{k}^k = \alpha^{p} \sum_{ l= j-k 	\vee 0 }^p  a_{p,l} \, b_{l,k,j} =  (A^p)_{k}^k \Big|_{\alpha=1} \alpha^p . 
\end{equation}

To present our main results---particularly those concerning the spectral moments of planar symplectic ensembles---we introduce some additional notation, following the conventions of~\cite{AEP22}.
Let 
\begin{equation} \label{skew norm}
r_k = 2 ( h_{2k+1}-c_{2k+1} h_{2k} )
\end{equation}
and
\begin{equation} \label{mu in sop}
 \mu_{k,j} = \prod_{l=j}^{k-1} \lambda_l, \qquad \lambda_l = \frac{h_{2l+2} - c_{2l+2} h_{2l+1}}{h_{2l+1} - c_{2l+1} h_{2l}}.
\end{equation} 
In terms of these quantities, we define 
\begin{equation} \label{def of mathfrak m jk}
\begin{split}
    \mathfrak{m}_{p_1,p_2,k} := & \sum_{n=k-\lfloor\frac{p_1}{2}\rfloor}^{k+\lfloor\frac{p_1}{2}\rfloor} \frac{r_n}{r_k} (B^{p_1})^{2n+1}_{2k+1} (B^{p_2})^{2n}_{2k} + \sum_{n=k-\lfloor\frac{p_2}{2}\rfloor}^{k+\lfloor\frac{p_2}{2}\rfloor} \frac{r_n}{r_k} (B^{p_1})^{2n}_{2k} (B^{p_2})^{2n+1}_{2k+1}
    \\
    & - \sum_{n=k-\lfloor\frac{p_1+1}{2}\rfloor}^{k+\lfloor\frac{p_1+1}{2}\rfloor} \frac{r_n}{r_k} (B^{p_1})^{2n}_{2k+1} (B^{p_2})^{2n+1}_{2k} - \sum_{n=k-\lfloor\frac{p_2+1}{2}\rfloor}^{k+\lfloor\frac{p_2+1}{2}\rfloor} \frac{r_n}{r_k} (B^{p_1})^{2n+1}_{2k} (B^{p_2})^{2n}_{2k+1}, 
\end{split}
\end{equation}
where
\begin{equation}
\begin{aligned} \label{B in A}
    & (B^p)^{2n+1}_{2k+1} = (A^p)^{2n+1}_{2k+1},
    && (B^p)^{2n}_{2k+1} = (A^p)^{2n}_{2k+1} - \lambda_n (A^p)^{2n+2}_{2k+1}, \\
    & (B^p)^{2n+1}_{2k} = \sum_{j=0}^{k} \mu_{k,j} (A^p)^{2n+1}_{2j},
    && (B^p)^{2n}_{2k} = \sum_{j=0}^{k} \mu_{k,j} (A^p)^{2n}_{2j} - \lambda_n \sum_{j=0}^{k} \mu_{k,j} (A^p)^{2n+2}_{2j}.
\end{aligned}
\end{equation}
Throughout the paper, we adopt the notation $a \wedge b := \min\{a,b\}$ and $a \vee b := \max\{a,b\}$.

\begin{thm}[\textbf{Evaluation of spectral moments}] \label{thm: main} For a given weight function $\omega$, suppose that the associated planar orthogonal polynomials satisfy the three-term recurrence relation \eqref{eq: recurrence}.  
\begin{itemize}
    \item \textup{\textbf{(Spectral moments of random normal matrix ensembles)}} We have 
    \begin{equation} \label{eq: main1}
     M^{ \mathbb{C} }_{p_1,p_2,N} = \sum_{k=0}^{N-1} \sum_{j=k- p_1 \wedge p_2}^{k+ p_1 \wedge p_2} \frac{h_j}{h_k} (A^{p_1})^j_k (A^{p_2})^j_k, 
    \end{equation}
where $ (A^p)^j_k$ is given by \eqref{eq: Tp}. 
In particular, for the holomorphic moments,
\begin{align} \label{holomorphic moments}
    M^{ \mathbb{C} }_{p,0,N} = \sum_{k=0}^{N-1} (A^p)^k_k.
\end{align}
    \item \textup{\textbf{(Spectral moments of planar symplectic ensembles)}} 
We have 
\begin{align} \label{eq: main2}
\begin{split}
    M^{ \mathbb{H} }_{p_1,p_2,N} = \frac{1}{2} \sum_{k=0}^{N-1} \mathfrak{m}_{p_1,p_2,k}, 
\end{split}
\end{align}
where $\mathfrak{m}_{p_1,p_2,k}$ is given by \eqref{def of mathfrak m jk}.  
In particular, for the holomorphic moments,
\begin{equation} \label{eq: holomorpic moments}
    M^{ \mathbb{H} }_{p,0,N} = \frac{1}{2} M^{ \mathbb{C} }_{p,0,2N} - \frac{1}{2} \sum_{j=0}^{N-1} \mu_{N,j} (A^p)^{2N}_{2j}. 
\end{equation}
\end{itemize}
\end{thm}

Since explicit formulas for the coefficients $(A^p)_k^j$ associated with the planar Hermite, Laguerre, and Gegenbauer polynomials are given in Proposition~\ref{Prop_Ap jk explicit}, Theorem~\ref{thm: main} immediately yields closed-form expressions for the corresponding spectral moments, albeit potentially lengthy. In particular, in Corollary~\ref{Cor_eGinibre moments}, we provide an explicit formula for the elliptic Ginibre ensembles.

\begin{rem}
In a paper ~\cite{MS11} by Mezzadri and Simm, the spectral moments of the GSE, LSE, and JSE are computed explicitly by exploiting the rank-one perturbation structure of the correlation kernel, previously developed in~\cite{AFNV00,Wi99}. These results can be recovered by taking the Hermitian limit of~\eqref{eq: holomorpic moments}.
\end{rem}


Combining the observation \eqref{Ap proportional} with Theorem~\ref{thm: main} and Lemma~\ref{lem: scaling}, we obtain the following corollary as an immediate consequence.

\begin{cor}[\textbf{Relation to Hermitian ensembles:}] \label{Cor_holo moments}
Suppose that the assumption of Lemma~\ref{lem: scaling} holds. Let $\omega_\R$ be the weight of the orthogonal polynomial $P_k$ on $\R$. We denote by $M^{ \R }_{p,N} \equiv M^{ \R }_{p,N}(\omega_{\R})$ the associated spectral moment of Hermitian unitary ensemble. 
Then we have 
\begin{equation}
  M^{ \mathbb{C} }_{p,0,N} =  \alpha^p  M^{ \R }_{p,N}. 
\end{equation}
In particular, we have
\begin{equation}
 M_{2p,0,N}^{ \rm{ H } , \mathbb{C}} =  \tau^p M_{2p,N}^{ \rm GUE }, \qquad   M_{p,0,N}^{ \rm {L} , \mathbb{C}} =  \tau^p M_{p,N}^{ \rm LUE }, \qquad    M_{2p,0,N}^{ \rm {G} , \mathbb{C}} =  \tau^p M_{2p,N}^{ \rm JUE }. 
\end{equation} 
\end{cor}

This corollary shows that the holomorphic spectral moments of the non-Hermitian ensembles coincide with those of their Hermitian limits, up to a simple scaling. A similar structure was also recently observed in \cite{SK24} in the context of the spectral form factor of the complex elliptic Ginibre ensemble.

\subsection{Elliptic law and non-Hermitian Marchenko-Pastur law}

Next, we discuss the large-$N$ behaviour of the spectral moments. 
The limiting empirical measure of the ensembles \eqref{Gibbs complex} and \eqref{Gibbs symplectic}, after suitable normalisation ensuring that the limiting support is a compact subset of the complex plane, has been identified in the literature for certain classes of weight functions.
In particular, for the planar Hermite weight $\omega^{ \rm H }$, the limiting measure is explicitly given by the elliptic law
\begin{equation} \label{def of ellipitc law}
\frac{1}{1-\tau^2} \mathbbm{1}_{ S }(z)\,dA(z),  \qquad 
 S :=  \Big\{ (x,y) \in \R^2: \Big( \frac{x}{1+\tau}\Big)^2 + \Big( \frac{y}{1-\tau} \Big)^2 \le 1 \Big\}. 
\end{equation}  
On the other hand, for the planar Laguerre weight, it was shown in \cite{ABK21} that the limiting measure is given by
\begin{equation} \label{eq: non-Hermitian MP law}
\frac{1}{1-\tau^2} \frac{1}{\sqrt{4|z|^2 + (1-\tau^2)^2 \alpha^2}} \mathbbm{1}_{ \hat{S} }(z)\, dA(z),  
\end{equation}
where  
\begin{equation}
    \hat{S} : = \Big\{ (x,y) \in \R^2: \Big( \frac{x-\tau(2+\alpha)}{(1+\tau^2)\sqrt{1+\alpha}} \Big)^2 + \Big( \frac{y}{(1-\tau^2)\sqrt{1+\alpha}} \Big)^2 \le 1 \Big\}.
\end{equation}
This is referred to as a non-Hermitian extension of the Marchenko-Pastur law. The limiting measure for the planar Gegenbauer weight has not been fully identified except for a special case involving the Chebyshev polynomials of the second kind, see \cite{NAKP20}.

\begin{thm}[\textbf{Large $N$-limit of spectral moments in the elliptic Ginibre and non-Hermitian Wishart ensembles}] \label{Thm_elliptic and MP}
We have the following.
\begin{itemize}
    \item[\textup (i)] 
    We have for elliptic Ginibre ensembles
\begin{equation} \label{elliptic leading order}
\lim_{ N \to \infty } \frac{ 1 }{ N^{ \frac{p_1+p_2}{2}+1 } } \,M^{\rm H, \mathbb{C} }_{p_1,p_2,N}=   \lim_{ N \to \infty }  \frac{1}{ 2^{ \frac{p_1+p_2}{2} } N^{\frac{p_1+p_2}{2}+1}} M^{ \rm{H} ,\mathbb{H}}_{p_1,p_2,N}
=    C_1(p_1,p_2),  
\end{equation}
where 
\begin{equation} \label{def of C1 p1 p2}
 C_1(p_1,p_2):= \frac{1}{  \frac{p_1+p_2}{2}+1  }  \sum_{r\in\mathcal{I}_{p_1 \wedge p_2}} \tau^{ \frac{p_1+p_2}{2} + r } \binom{p_1}{ \frac{p_1+r}{2} } \binom{p_2}{ \frac{p_2+r}{2} }.
\end{equation}
Here,
\begin{equation} \label{I_p def}
        \mathcal{I}_{p} := \{ -p, -p+2, -p+4, \cdots, p-2, p \}.
    \end{equation} 
   \item[\textup (ii)] Suppose that  
\begin{equation}
    \lim_{ N \to \infty } \frac{\nu}{N} = \alpha\ge0.
\end{equation} 
Then we have for non-Hermitian Wishart ensembles
    \begin{equation} \label{shifted elliptic leading order}
    \lim_{N\to\infty} \frac{1}{N^{p_1+p_2+1}} M^{\rm L, \mathbb{C} }_{p_1,p_2,N}  =  \lim_{N\to\infty} \frac{1}{2^{p_1+p_2} N^{p_1+p_2+1}} M^{\rm L, \mathbb{H} }_{p_1,p_2,N}  =   L_1(p_1,p_2) ,  
    \end{equation}
    where 
    \begin{equation} \label{def of L1 p1 p2}
     L_1(p_1,p_2) := \sum_{r=-p_1 \wedge p_2}^{p_1 \wedge p_2} \sum_{l_1=0}^{p_1} \sum_{l_2=0}^{p_2} \tau^{ p_1+p_2+2r } \frac{\alpha^{p_1+p_2-l_1-l_2}}{l_1+l_2+1} \binom{p_1}{l_1} \binom{p_1+l_1}{l_1+r} \binom{p_2}{l_2} \binom{p_2+l_2}{l_2-r}.
    \end{equation}
\end{itemize}
\end{thm}

\begin{rem}
The leading-order spectral moments can be used to derive the limiting spectral distribution via the conformal mapping method, particularly when analysing the Cauchy transform of the equilibrium measure, see \cite{ABK21} and \cite[Remark~2.5]{By24}. In connection with the elliptic law \eqref{def of ellipitc law}, it follows that
\begin{equation} \label{eq_elliptic law moment}
\frac{1}{1-\tau^2} \int_S z^{p_1} \overline{z}^{p_2} \, dA(z) 
= C_1(p_1, p_2).
\end{equation}
In Appendix~\ref{Appendix_moment elliptic}, we provide a direct computation of this identity using the conformal map and the Schwarz function.

Similarly, by virtue of the non-Hermitian Marchenko--Pastur law \eqref{eq: non-Hermitian MP law}, we obtain
\begin{equation} \label{eq_non-Hermitian MP law moment}
\frac{1}{1-\tau^2} \int_{\hat{S}} \frac{z^{p_1} \overline{z}^{p_2}}{\sqrt{4|z|^2 + (1 - \tau^2)^2 \alpha^2}} \, dA(z) = L_1(p_1, p_2).
\end{equation}
Unlike in the case of the elliptic law, a similar computation in Appendix~\ref{Appendix_moment elliptic} yields a more intricate expression for this integral. This, in turn, highlights that computing spectral moments provides a simpler means of evaluating such integrals.


In Theorem~\ref{Thm_elliptic and MP}, we present results only for the elliptic Ginibre and non-Hermitian Wishart ensembles, and not for the planar Gegenbauer ensemble. In principle, analogous computations can be carried out using similar algebraic manipulations; however, this leads to particularly lengthy and complicated formulas in the Gegenbauer case. 
Indeed, unlike the former two models, the limiting global measure for the planar Gegenbauer ensemble has not yet been established in the literature; this will be addressed in forthcoming work. As a consequence, one cannot perform the same consistency checks as in \eqref{eq_elliptic law moment} and \eqref{eq_non-Hermitian MP law moment} for the former two models, which is one of the reasons we do not include the corresponding formulas here. 
\end{rem}

\begin{ex}
The explicit formulas of $C_1(p_1,p_2)$ for the first few values of $p_1$ and $p_2$ are as follows. 
\begin{align*} 
    & C_1(0,0) = 1 && C_1(2,0) = \tau, && C_1(1,1) = \tfrac{1}{2}\tau^2+\tfrac{1}{2}, \\
    & C_1(4,0) = 2\tau^2, && C_1(3,1) = \tau^3+\tau, && C_1(2,2) = \tfrac{1}{3}\tau^4+\tfrac{4}{3}\tau^2+\tfrac{1}{3}, \\
    & C_1(6,0) = 5\tau^3, && C_1(5,1) = \tfrac{5}{2}\tau^4+\tfrac{5}{2}\tau^3 && C_1(4,2) = \tau^5+3\tau^3+\tau, && C_1(3,3) = \tfrac{1}{4}\tau^6+\tfrac{9}{4}\tau^4+\tfrac{9}{4}\tau^2+\tfrac{1}{4}, \\
    & C_1(8,0) = 14\tau^4, && C_1(7,1) = 7\tau^5+7\tau^3, && C_1(6,2) = 3\tau^6+8\tau^4+3\tau^2, && C_1(5,3) = \tau^7+6\tau^5+6\tau^3+\tau.
\end{align*}
Observe here that as a polynomials in $\tau$, the sum of coefficients is same as the $p$-th Catalan number $C_p=\frac{1}{p+1}\binom{2p}{p}$ with $p=(p_1+p_2)/2$. 

In addition, explicit formulas for $L_1(p_1, p_2)$ for the first few values of $p_1$ and $p_2$ are provided below, where we set $\gamma = 1 + \alpha$.
\begin{align*}
    & L_1(0,0) = 1 && L_1(1,0) = \gamma\tau, \\
    & L_1(2,0) =  (\gamma^2+\gamma )\tau, && L_1(1,1) = (\tfrac{1}{2}\gamma-\tfrac{1}{6})\tau^4 + (\gamma^2+\tfrac{1}{3})\tau^2 + (\tfrac{1}{2}\gamma-\tfrac{1}{6}), 
    \\
    &L_1(3,0) = (\gamma^3+3\gamma+\gamma)\tau^3 && L_2(2,1) = \gamma^2\tau^5 + (\gamma^3+\gamma^2+\gamma)\tau^3 + \gamma^2\tau, \\
    &L_1(4,0) =  (\gamma^4+6\gamma^3+6\gamma+\gamma )\tau^4 && L_2(3,1) = (\tfrac{3}{2}\gamma^3+\tfrac{3}{2}\gamma )\tau^6 + (\gamma^4+3\gamma^3+3\gamma^2+\gamma)\tau^4 + (\tfrac{3}{2}\gamma^3+\tfrac{3}{2}\gamma )\tau^2.
\end{align*}
\end{ex}

\begin{rem}
We now consider the Hermitian limit $\tau \to 1$, in which $C_1$ and $L_1$ converge to moments of the semicircle and Marchenko--Pastur laws corresponding to the GUE and LUE, respectively.
\begin{itemize}
    \item  Without loss of generality, suppose $p_1\ge p_2.$ Let $p=(p_1+p_2)/2.$ Then we have 
\begin{equation}
\begin{split} 
   C_1(p_1,p_2)\Big|_{ \tau=1 }  
    & = \frac{1}{p+1} \sum_{s=0}^{p_1} \binom{p_1}{ s } \binom{p_2}{p-s}
   = \frac{1}{p+1} \binom{2p}{p}  = C_p, 
\end{split}
\end{equation}
where $C_p$ is the $p$-th Catalan number. 
   \smallskip 
   \item Let $p=p_1+p_2.$ In the Hermitian limit, we have
\begin{equation} \label{L1 p1p2 tau1 Narayana}
 L_1(p_1,p_2)  \Big|_{ \tau=1 } = N_p(1+\alpha) := \sum_{k=1}^{p} \frac{1}{p}\binom{p}{k}\binom{p}{k-1} (1+\alpha)^k, 
\end{equation}
where $N_p$ is called the Narayana polynomial \cite{MS08}, see Appendix~\ref{Appendix_combinatorial identities}. 
\end{itemize}
\end{rem}

\begin{rem}
When $\alpha = 0$, we have
\begin{equation}
    L_1(p_1,p_2) \big|_{\alpha=0} = \frac{1}{p_1 + p_2 + 1} \sum_{r = -p_1 \wedge p_2}^{p_1 \wedge p_2} \tau^{p_1 + p_2 + 2r} \binom{2p_1}{p_1 + r} \binom{2p_2}{p_2 - r} = C_1(2p_1, 2p_2).
\end{equation}
This is consistent with the fact that the non-Hermitian Marchenko--Pastur law~\eqref{eq: non-Hermitian MP law} with $\alpha = 0$ coincides with the elliptic law~\eqref{def of ellipitc law} under the transformation $z \mapsto z^2$, see~\cite{ABK21}.
\end{rem}

\subsection{Spectral moments of the complex and symplectic elliptic Ginibre ensemble}

For the elliptic Ginibre ensembles, it was shown in \cite{LR16} for the complex case and in \cite{BE23} for the symplectic case that the associated correlation kernels exhibit properties analogous to the classical Christoffel--Darboux formula. More precisely, after applying an appropriate differential operator, the correlation kernel of the complex ensemble—forming a determinantal point process—can be expressed in terms of a few orthogonal polynomials of the highest degrees. Furthermore, when a similar differential operator is applied to the kernel of the symplectic ensemble—which forms a Pfaffian point process—it can be written in terms of the complex kernel, supplemented by a correction term. This structure is reminiscent of the rank-one perturbation relation between the GSE and GUE correlation kernels; see \cite{AFNV00,Wi99,Bo25}.

These structural properties allow us to derive alternative formulas for the spectral moments, which share similar features with those in Corollary~\ref{Cor_eGinibre moments}, while offering additional advantages for the asymptotic analysis.

\begin{thm}[\textbf{Spectral moments of the complex and symplectic elliptic Ginibre ensemble}] \label{Thm_eGinibre 2}
Let $(A^p)_k^j$ denote the coefficients defined in~\eqref{eq: Tp} corresponding to the Hermite weight $\omega^{\rm H}$, whose explicit expression is given in~\eqref{Ap Hermite}.
\begin{itemize}
    \item[\textup (i)] We have for the complex ensemble
    \begin{equation} \label{Thm 2.5_1}
    M^{ \rm{H} ,\mathbb{C}}_{p_1,p_2,N} = \frac{1}{1-\tau^2} \frac{1}{p_1+1} \sum_{n=N-1-(p_1+1)\vee p_2}^{N+(p_1+1)\vee p_2} \frac{n!}{(N-1)!} \Big[ (A^{p_1+1})^n_{N-1} (A^{p_2})^n_N - \tau (A^{p_1+1})^n_N (A^{p_2})^n_{N-1} \Big].
\end{equation}
    \item[\textup (ii)] We have for the symplectic ensemble
\begin{equation} 
\begin{split}  \label{Thm 2.5_2}
    M^{\rm{H},\mathbb{H}}_{p_1,p_2,N}
    & = \frac{1}{2} M^{ \rm{H} ,\mathbb{C}}_{p_1,p_2,2N} + (1-\tau^2) \frac{ p_1 p_2 }{ p_1 + p_2 } M^{ \rm{H} ,\mathbb{H}}_{p_1-1,p_2-1,N} 
    \\
    & \quad - \frac{1}{2 } \sum_{k=0}^{N-1} \sum_{n=2N - p_1 \vee p_2}^{2N + p_1 \vee p_2} \frac{n!}{(2N)!} \frac{(2N)!!}{(2k)!!} \bigg[ \frac{p_1}{p_1+p_2} (A^{p_1})^n_{2k} (A^{p_2})^n_{2N} + \frac{p_2}{p_1+p_2} (A^{p_1})^n_{2N} (A^{p_2})^n_{2k} \bigg].
\end{split}
\end{equation}
\end{itemize}
\end{thm}

As an immediate consequence, we have the following. 

\begin{ex}
We consider the two extremal cases; the Gaussian ensemble and the Ginibre ensemble.
\begin{itemize}

\item $(\tau=0)$. We have for the spectral moments for the GinUE and GinSE: 
\begin{align} \label{eq: ginue moment}
    & M^{\text{GinUE}}_{p_1,p_2,N} =  \frac{1}{p_1+1} \frac{(N+p_1)!}{(N-1)!} \mathbbm{1}_{ \{ p_1=p_2=p \} } 
    \\ 
    \label{eq: ginse moment}
    & M^{\text{GinSE}}_{p_1,p_2,N} = \begin{cases}
         \displaystyle   \sum_{k=0}^{N-1} \frac{(2k+1+p)!}{(2k+1)!} & \text{if } p_1=p_2=p,
         \smallskip 
         \\
       \displaystyle  -\frac{p_1}{2 } \sum_{k=0}^{N-1} \frac{(2k+1+p_2)!}{(2k+1)!} \frac{(2k)!!}{(2k+2-p_1+p_2)!!} & \text{if } p_1>p_2.
    \end{cases}
\end{align} 
(The moment in \eqref{eq: ginue moment} also arises in the study of random permutations, see e.g. \cite{Dub24}.)
Note that the particular case $(p_1,p_2)=(2p,0)$ of \eqref{eq: ginse moment} becomes
\begin{equation}
     M^{\text{GinSE}}_{2p,0,N} = -2^{p-1}p \sum_{k=0}^{N-1} \frac{k!}{(k+1-p)!} = -2^{p-1} \sum_{k=0}^{N-1} \Big\{ \frac{(k+1)!}{(k+1-p)!} - \frac{k!}{(k-p)!} \Big\} = -2^{p-1} \frac{N!}{(N-p)!}.
\end{equation}
This formula agrees with that given in \cite[Corollary 4.6]{Eb21}.

\smallskip

\item $(\tau\to1)$. We have for the spectral moments of the GUE and GSE:
\begin{equation} \label{eq: Gaussian moments}
\begin{split}
    & M^{\text{GUE}}_{2p,N} = M^{\rm{H} ,\mathbb{C}}_{2p,0,N} \Big\vert_{\tau=1} = \sum_{l=0}^p \frac{(2p)!}{2^l l! (p-l)!} \binom{N}{p-l+1},
    \\
    & M^{\text{GSE}}_{2p,N} = M^{\rm{H},\mathbb{H}}_{2p,0,N} \Big\vert_{\tau=1} = \frac{1}{ 2 } M^{\text{GUE}}_{2p,2N} - \frac{1}{ 2 } \sum_{r=1}^{p} \sum_{l=0}^{p} \frac{(2N)!!}{(2N-2r)!!} \frac{(2p)!}{2^l l! (p-l+r)!} \binom{2N-2r}{p-l-r}.
\end{split}
\end{equation}
This formula is equivalent to \cite[Theorem 2.9]{MS11}.
\end{itemize}
\end{ex} 

As previously mentioned, the spectral moments of the GUE has been extensively studied in the literature. In particular, it is well known that it admits a large-$N$ expansion of the form \cite{BFO24} 
\begin{equation}\label{GUE genus expansion}
\frac{1}{N^{p+1} } M_{N,2p}^{ \rm GUE } =    \sum_{g=0}^{\lfloor (p+1)/2 \rfloor}  \frac{ \mathcal{E}_g(p) }{ N^{2g} } ,  \qquad  \mathcal{E}_g(p) :=    (2p-1)!! \sum_{m=0}^{2g}  \frac{ s(p+1-m,p+1-2g)}{(p+1-m)!}   \binom{p}{m} 2^{p-m}, 
\end{equation}
where $s(n,k)$ are the Stirling numbers of the first kind \cite[Section 26.8]{NIST}.
We remark that the coefficient $\mathcal{E}_g(p)$ enumerates the number of pairings of the edges of a $2p$-gon, which is dual to a map on a compact Riemann surface of genus $g$ (cf.~\cite{HZ86}). For this reason, the expansion \eqref{GUE genus expansion} is commonly referred to as the \emph{genus expansion}.

By Corollary~\ref{Cor_holo moments}, the genus expansion \eqref{GUE genus expansion} also yields a $1/N^2$ expansion for the holomorphic spectral moments of the complex elliptic Ginibre ensemble. In contrast, for other cases—such as mixed moments and symplectic ensembles—a separate asymptotic analysis is required.

As an application of Theorem~\ref{Thm_eGinibre 2}, we have the large-$N$ expansion of the spectral moments. 

\begin{thm}[\textbf{Genus type expansion of the elliptic Ginibre ensemble}] \label{Thm_eGinUSE genus expansion}
As $N\to \infty$, we have the following.
\begin{itemize}
    \item[\textup (i)] We have for the complex ensemble \begin{equation}
    \frac{1}{N^{\frac{p_1+p_2}{2}+1}} M^{ \rm{H} ,\mathbb{C}}_{p_1,p_2,N} = C_1(p_1,p_2) + C_2(p_1,p_2) \frac{1}{N} + O(N^{-2}),
\end{equation}
where $C_1(p_1,p_2)$ is given by \eqref{def of C1 p1 p2}, and 
\begin{equation} \label{def of C2 p1 p2}
\begin{split}
 C_2(p_1,p_2) := -\sum_{r\in\mathcal{I}_{p_1 \wedge p_2}} \tau^{\frac{p_1+p_2}{2}+r} \binom{p_1}{\frac{p_1+r}{2}} \binom{p_2}{\frac{p_2+r}{2}} \frac{r}{2}.
\end{split}
\end{equation}
     \item[\textup (ii)] We have for the symplectic ensemble
\begin{equation}
    \frac{1}{ 2^{ \frac{p_1+p_2}{2} } N^{\frac{p_1+p_2}{2}+1}} M^{ \rm{H} ,\mathbb{H}}_{p_1,p_2,N} = C_1(p_1,p_2) + C_2^{\prime}(p_1,p_2) \frac{1}{N} + O(N^{-2}),
\end{equation}
where
\begin{equation} \label{eq: C2prime}
\begin{split}
   &\quad C_2^{\prime}(p_1,p_2)  := \frac{1}{2} C_2(p_1,p_2) + \frac{1-\tau^2}{2} \frac{p_1p_2}{p_1+p_2} C_1(p_1-1,p_2-1)
    \\
    & \quad - \frac{1}{2} \sum_{r\in\mathcal{I}_{p_1 \vee p_2}} \sum_{s=1}^{p_1 \vee p_2} \tau^{\frac{p_1+p_2}{2}+r-s} \bigg[ \frac{p_1}{p_1+p_2} \binom{p_1}{\frac{p_1+r}{2}-s} \binom{p_2}{\frac{p_2+r}{2}} + \frac{p_2}{p_1+p_2} \binom{p_1}{\frac{p_1+r}{2}} \binom{p_2}{\frac{p_2+r}{2}-s} \bigg].
\end{split}
\end{equation} 
\end{itemize}
\end{thm}

We remark that, unlike the holomorphic moments—where $M^{\rm{H},\mathbb{C}}_{p_1,p_2,N}$ admits a $1/N^2$ expansion—the general mixed moments typically exhibit a $1/N$ expansion.

\begin{ex}
We consider the two extremal cases. 
\begin{itemize}
    \item ($\tau\to1$). In the Hermitian limit, we have
    \begin{equation}  \label{C2 prime tau1}
        C_2(p_1,p_2) \Big\vert_{\tau\to1} = 0, \qquad 
     C'_2(p_1,p_2) \Big\vert_{\tau\to1} =  -\frac{1}{2} \sum_{l=0}^{p-1} \binom{2p}{l},
    \end{equation}
    where $p=(p_1+p_2)/2$. See Appendix~\ref{Appendix_combinatorial identities} for a verification of $C_2$. 
    This is consistent with the known fact that the spectral moments of the GUE admit a $1/N^2$ expansion, whereas those of the GSE admit a $1/N$ expansion, see, e.g. \cite{WF14}.
  Explicit formulas for the expansion of the GSE spectral moments with small values of $p$ are provided in~\cite[Theorem 6]{Le09} and~\cite[Eq.~(3.27)–(3.33)]{WF14}.
    \smallskip 
    \item ($\tau=0$). In this case, we have  
    \begin{equation}
        C_2(p_1,p_2) \Big\vert_{\tau=0} = \frac{p_1}{2} \mathbbm{1}_{\{ p_1=p_2 \} },    \qquad    C'_2(p_1,p_2) \Big\vert_{\tau=0} = \frac{1}{4}(p+1) \mathbbm{1}_{\{p_1=p_2=p\} } - \frac{1}{2} \frac{p_1\vee p_2}{p_1+p_2} \mathbbm{1}_{\{p_1\ne p_2\}}.
    \end{equation}
   This can also be directly checked using \eqref{eq: ginue moment} and \eqref{eq: ginse moment}.
\end{itemize}
\end{ex}

\begin{rem}
The differential operator approach used in Theorems~\ref{Thm_eGinibre 2} and~\ref{Thm_eGinUSE genus expansion} is based on \cite[Proposition 2.3]{LR16} for the complex elliptic ensemble and \cite[Proposition 1.1]{BE23} for the symplectic elliptic ensemble. In particular, by applying a suitable integration by parts argument, we derive Theorem~\ref{Thm_eGinibre 2}. A key advantage in this setting is that the associated differential operator is of first order, together with the simple exponential form of the weight function \eqref{def of weight H}, both of which facilitate the computation.
An analogous differential operator formula for the complex and symplectic non-Hermitian Wishart ensembles was obtained in the recent work \cite[Theorem 1.1]{BN24}. However, in contrast to the elliptic case, the non-Hermitian Wishart case involves a second-order differential operator. Combined with the more complicated form of the weight function \eqref{def of weight L}, which involves the modified Bessel function, this makes the method—particularly when applying integration by parts—significantly more difficult to implement. 
\end{rem}

\subsection*{Plan of the paper} The remainder of this paper is organised as follows. In Section~\ref{Section_Preliminaries}, we present some basic preliminaries, including fundamental properties and explicit formulas for classical orthogonal polynomials, as well as the integrable structure of Coulomb gas ensembles. 
Section~\ref{Section_Proofs} is devoted to the proofs of our main results, namely Theorems~\ref{thm: main}, \ref{Thm_elliptic and MP}, \ref{Thm_eGinibre 2}, and \ref{Thm_eGinUSE genus expansion}. 
Several appendices are included. In Appendix~\ref{Appendix_moment elliptic}, we derive the mixed moments of the elliptic law. In Appendix~\ref{Appendix_eGinibre spectral moments}, we provide an explicit formula for the spectral moments of elliptic Ginibre matrices as a consequence of Theorem~\ref{thm: main}. Finally, Appendix~\ref{Appendix_combinatorial identities} contains verifications of several elementary combinatorial identities.

\subsection*{Acknowledgements} We thank Peter J. Forrester for helpful discussions and for his comments on the draft of the manuscript.

\section{Preliminaries} \label{Section_Preliminaries} 

This section is devoted to compiling basic properties of classical orthogonal polynomials, together with several explicit formulas that will be used throughout the paper. We also recall key integrable features of the two-dimensional ensembles \eqref{Gibbs complex} and \eqref{Gibbs symplectic}, and describe the role of planar (skew)-orthogonal polynomials in their analysis.

We begin with a brief proof of Lemma~\ref{lem: scaling}.

\begin{proof}[Proof of Lemma~\ref{lem: scaling}] 
Combining \eqref{inverse representation} and \eqref{linearisation}, we have 
\begin{equation*}
x^p P_j(x) = \sum_{l=0}^p  a_{p,l}  P_l(x) P_j(x) = \sum_{l=0}^p a_{p,l}\sum_{m=0}^{l+j} b_{l,j,m} P_m(x). 
\end{equation*} 
Then the lemma follows from  
\begin{align*}
T_p \, p_k(z) &=  \alpha^{p+k} \  T_p P_k\Big( \frac{z}{\alpha}\Big)  = \alpha^{p+k} \sum_{l=0}^p a_{p,l}\sum_{m=0}^{l+k} b_{l,k,m} P_m\Big( \frac{z}{\alpha}\Big) =  \sum_{l=0}^p a_{p,l}\sum_{m=0}^{l+k} b_{l,k,m} \alpha^{p+k-m}  p_m(z)
\\
&= \sum_{ j=k-p }^{ k+p }  \sum_{ l= j-k 	\vee 0 }^p  a_{p,l} b_{l,k,j} \alpha^{p+k-j} p_j(z) . 
\end{align*} 
This completes the proof. 
\end{proof}

\subsection{Hermite, Laguerre and Gegenbauer polynomials} \label{Subsec_OP basics}

We recall the definitions of the Hermite, Laguerre, and Gegenbauer polynomials, and present their inversion and linearisation coefficients in Definition~\ref{Def_mixed and linearisation}.

The  Hermite, generalised Laguerre, and Gegenbauer polynomials are defined by \cite{NIST}
\begin{align}
 \label{def of Hermite}
H_k(x) & := (-1)^k e^{x^2} \frac{d^n}{dx^n} e^{-x^2},
\\
 \label{def of Laguerre}
L^\nu_k(x) &:= \frac{x^{-\nu}e^x}{k!} \frac{d^k}{dx^k} (x^{k+\nu} e^{-x}),
\\
 \label{def of Gegenbauer}
C_n^{(a)}(x) &:= \frac{(-1)^n}{2^n n!} \frac{\Gamma(a + \tfrac{1}{2}) \Gamma(n + 2a)}{\Gamma(2a) \Gamma(a + n + \tfrac{1}{2})}
(1 - x^2)^{-a + 1/2} \frac{d^n}{dx^n} \Big[ (1 - x^2)^{n + a - 1/2} \Big]. 
\end{align} 
They form families of orthogonal polynomials on the real line with respect to the weights \eqref{def of weight H Hermitian}, \eqref{def of weight L Hermitian}, and \eqref{def of weight G Hermitian}, respectively. As previously mentioned, they also define planar orthogonal polynomial systems as in \eqref{def of planar Hermite}, \eqref{def of planar Laguerre}, and \eqref{def of planar Gegenbauer}.

Their inversion and linearisation coefficients are given as follows. 
\begin{itemize}
    \item Inversion coefficients (cf. \cite[Theorem 5]{KS98}):
\begin{align}
\label{hermite inversion}
&x^{n}  = \frac{n!}{2^{n} } \sum_{l=0}^{ \lfloor n/2 \rfloor } \frac{1}{ l!( n-2l)! }H_{n-2l}(x) , 
\\ \label{laguerre inversion}
& x^n  = n! \sum_{l=0}^n (-1)^l \binom{n+\nu}{n-l} L_l^{\nu}(x), 
\\
\label{gegenbauer inversion}
& x^n  = \frac{n!}{2^n} \sum_{l=0}^{\lfloor n/2 \rfloor} \frac{n+a-2l}{l!\ (a)_{n+1-l}} C^{(a)}_{n-2l}(x).
\end{align} 
\item Linearisation coefficients (cf. \cite[Eqs.(3.18),(3.20)]{Ze92} and \cite[Eq.(26)]{Ch10}):  
\begin{align} \label{hermite linearisation} 
 H_n(x) H_m(x)&= \sum_{s=0}^{\lfloor (n+m)/2 \rfloor}  2^s\, s! \binom{n}{s} \binom{m}{s}    H_{n+m-2s}(x),
 \\
 \label{laguerre linearisation}
L_n^{\nu}(x) L_m^{\nu}(x) & = \sum_{k=0}^{n+m} \bigg( \sum_{ s \ge 0 } \frac{   (-2)^{k+n+m-2s} k! (\nu+s)! }{ (k+\nu)! (s-k)!(s-n)!(s-m)! (k+n+m-2s)! } \bigg) L_k^{\nu}(x), 
\\
\label{gegenbauer linearisation}
 C^{(a)}_n(x) C^{(a)}_m(x)  & = \sum_{l=0}^{n\wedge m} \frac{n+m+a-2l}{n+m+a-l} \frac{(a)_l (a)_{n-l} (a)_{m-l} (2a)_{n+m-l} (n+m-2l)!}{l! (n-l)! (m-l)! (a)_{n+m-l} (2a)_{n+m-2l}} C^{(a)}_{n+m-2l}(x). 
\end{align}
\end{itemize}
Note that the summation over $s\ge0$ in \eqref{laguerre linearisation} becomes finite if we adopt the convention that the reciprocal factorial \( \frac{1}{n!} \) is interpreted as \( \frac{1}{\Gamma(n+1)} = 0 \) for negative integers \( n \).

The classical three-term recurrence relations for the Hermite, Laguerre, and Gegenbauer polynomials imply that the corresponding planar orthogonal polynomials \eqref{def of planar Hermite}, \eqref{def of planar Laguerre}, and \eqref{def of planar Gegenbauer} satisfy
\begin{align}
 \label{hermite recurrence}
    z \, p_{k}^{\rm H}(z) & = p_{k+1}^{\rm H}(z) + k\tau \, p_{k-1}^{\rm H}(z),
    \\
    \label{laguerre recurrence}
    z \, p_k^{\rm L}(z) & = p_{k+1}^{\rm L}(z) + \tau(2k+1+\nu) \, p_{k}^{\rm L}(z) + \tau^2 k(k+\nu) \, p_{k-1}^{\rm L}(z), 
    \\
    \label{gegenbauer recurrence}
    z \, p_k^{\rm G}(z) & = p_{k+1}^{\rm G}(z) + \frac{\tau}{4} \frac{k(k+1+2a)}{(k+a)(k+1+a)} \, p_{k-1}^{\rm G}(z).
\end{align}   
These relations define the coefficients $b_k$ and $c_k$ in \eqref{eq: recurrence}.
Recall that the coefficients $(A^p)_k^j$ are defined by \eqref{eq: Tp}. Below, we provide their explicit formulas for the planar Hermite, Laguerre, and Gegenbauer polynomials.


\begin{prop} \label{Prop_Ap jk explicit}
Recall that $\mathcal{I}_p$ is given by \eqref{I_p def}. 
\begin{itemize}
    \item For the planar Hermite polynomial \eqref{def of planar Hermite}, we have  
    \begin{equation} \label{Ap Hermite}
     (A^p)^j_k =  \tau^{\frac{p+k-j}{2}} \sum_{l=0}^{\lfloor p/2 \rfloor} \frac{p!}{2^l \, l!(\frac{p-k+j}{2}-l)!} \binom{k}{\frac{p+k-j}{2}-l}
    \end{equation}
    if $ j-k\in\mathcal{I}_{p},$ and $  (A^p)^j_k =0$ otherwise.  
    \smallskip 
     \item For the planar Laguerre polynomial \eqref{def of planar Laguerre}, we have  
     \begin{align} 
    \begin{split}  \label{Ap laguerre}
     (A^p)^j_k &=  \tau^{p+k-j} \sum_{l=0}^p \frac{p!}{(p-l)!} \frac{\Gamma(p+\nu+1)}{\Gamma(l+\nu+1)} \frac{k!}{\Gamma(j+\nu+1)} 
 \sum_{s\ge0} \frac{\Gamma(s+\nu+1)\,2^{j+k+l-2s}}{(s-j)!(s-k)!(s-l)!(j+k+l-2s)!} 
    \end{split}
     \end{align}
     if $\abs{j-k}\le p$, and $  (A^p)^j_k =0$ otherwise.  
    \smallskip 
     \item  For the planar Gegenbauer polynomial \eqref{def of planar Gegenbauer}, we have
     \begin{align}
    \begin{split}
     (A^p)^j_k & =   \tau^{\frac{p+k-j}{2}} \sum_{l=0}^{\lfloor p/2 \rfloor} \frac{k!\ 2^j (1+a)_j}{2^k (1+a)_k} \frac{p!}{2^p} \frac{p-2l+a+1}{l!\ (a+1)_{p+1-l}} \frac{j+a+1}{(p-2l+k+j)/2+a+1} 
        \\
        & \qquad \times \frac{  (a+1)_{\frac{k+p-2l-j}{2}} (a+1)_{\frac{j+p-2l-k}{2}} (a+1)_{\frac{j+k-p+2l}{2}} (2a+2)_{\frac{j+k+p-2l}{2}} } { \big(\frac{k+p-2l-j}{2}\big)! \big(\frac{j+p-2l-k}{2}\big)! \big(\frac{j+k-p+2l}{2}\big)! (a+1)_{\frac{j+k+p-2l}{2}} (2a+2)_{j} }
    \end{split}
     \end{align}
     if $j-k\in\mathcal{I}_{p},$ and $  (A^p)^j_k =0$ otherwise.  
\end{itemize}
\end{prop}

\subsection{Integrable structure of complex and symplectic ensembles}

In this subsection, we recall the integrable structure of \eqref{Gibbs complex} and \eqref{Gibbs symplectic}. 

The $k$-point correlation functions of the ensembles \eqref{Gibbs complex} and \eqref{Gibbs symplectic} are defined by 
\begin{equation}
\begin{split} \label{def of RNk}
& R_{N,k}^{ \mathbb{C} }(z) = \frac{N!}{(N-k)!} \frac{1}{ Z_N^{ \mathbb{C} } } \int_{ \C^{ N-k } } \P_N^{ \mathbb{C} }( z_1,z_2,\dots,z_N ) \,dA(z_{k+1}) \dots dA(z_N), 
\\
& R_{N,k}^{ \mathbb{H} }(z) = \frac{N!}{(N-k)!} \frac{1}{ Z_N^{ \mathbb{H} } } \int_{ \C^{ N-k } } \P_N^{ \mathbb{H} }( z_1,z_2,\dots,z_N ) \,dA(z_{k+1}) \dots dA(z_N). 
\end{split}
\end{equation}
Then by definition, the spectral moments $M^{ \mathbb{C} }_{p_1,p_2,N}$ and $M^{ \mathbb{H} }_{p_1,p_2,N}$ can be written in terms of the $1$-point function:
\begin{equation} \label{moments r1}
M^{ \mathbb{C} }_{p_1,p_2,N} = \int_\C z^{p_1} \overline{z}^{p_2} R_{N,1}^{ \mathbb{C} }(z)\,dA(z), \qquad M^{ \mathbb{H} }_{p_1,p_2,N} = \int_\C z^{p_1} \overline{z}^{p_2} R_{N,1}^{ \mathbb{H} }(z)\,dA(z). 
\end{equation}

It is well known that the ensembles \eqref{Gibbs complex} and \eqref{Gibbs symplectic} form determinantal and Pfaffian point processes, respectively, whose correlation kernels are expressed in terms of the associated planar orthogonal and skew-orthogonal polynomials.
Recall the inner product $\langle \cdot , \cdot \rangle $ is given by \eqref{inner product} and $(p_k)_{k=1}^\infty$ is the associated planar orthogonal polynomials \eqref{OP norm}. 
In addition, we define a skew-symmetric form on $\mathbb{R}[z]$: 
\begin{equation}
\label{skew product}
\langle f, g\rangle_s : =\int_{\mathbb{C}} \Big( f(z)\overline{g(z)}-g(z)\overline{f(z)} \Big)(\overline{z}-z) \, \omega(z)\,dA(z). 
\end{equation}
A family of polynomials $(q_k)_{k=1}^\infty$ is called \textit{planar skew-orthogonal polynomials} if it satisfies  
 \begin{equation} \label{skew orthogonality}
    \langle q_{2j},q_{2k} \rangle_s=\langle q_{2j+1},q_{2k+1} \rangle_s=0,\qquad 
    \langle q_{2j+1},q_{2k} \rangle_s=-\langle q_{2j},q_{2k+1} \rangle_s=r_k \, \delta_{j,k},
\end{equation}
where $r_k$ is their skew-norm.
Using $(p_k)_{k=1}^\infty$ and $(q_k)_{k=1}^\infty$, we define
\begin{align}
 \widehat{K}^\mathbb{C}_N(z,w) := \sum_{k=0}^{N-1} \frac{1}{h_k} p_k(z) p_k(w), \qquad
    & \widehat{K}^\mathbb{H}_N(z,w) := \sum_{k=0}^{N-1} \frac{1}{r_k} \Big( q_{2k+1}(z) q_{2k}(w) - q_{2k}(z) q_{2k+1}(w) \Big),
\end{align}
and 
\begin{equation} \label{def of kernel}
\begin{split}
K^\mathbb{C}_N(z,w) := \sqrt{\omega(z)\omega(w)} \, \widehat{K}^\mathbb{C}_N(z,w), \qquad K_N^\mathbb{H}(z,w) := \sqrt{\omega(z)\omega(w)} \, \widehat{K}_N^\mathbb{H}(z,w). 
\end{split}
\end{equation}
Then, it is well known that (see e.g. \cite{BF25}) the $k$-point correlation functions \eqref{def of RNk} can be written as 
\begin{equation} \label{k-point}
\begin{split}
R_{N,k}^{ \mathbb{C} }(z) = \det \Big[ K^\mathbb{C}_N(z_j,z_l) \Big]_{j,l=1}^k, 
    \qquad 
     R_{N,k}^{ \mathbb{H} }(z) = \prod_{j=1}^k (\overline{z}_j-z_j) \, \Pf \begin{bmatrix}
        K^\mathbb{H}_N(z_j,z_l) & K^\mathbb{H}_N(z_j,\overline{z}_l) 
        \smallskip 
        \\
        K^\mathbb{H}_N(\overline{z}_j,z_l) & K^\mathbb{H}_N(\overline{z}_j,\overline{z_l})
    \end{bmatrix}_{j,l=1}^k.
\end{split}
\end{equation}
Note in particular that the $1$-point functions can be written as  
\begin{equation} \label{1-point}
    \begin{split}
    & R_{N,1}^{ \mathbb{C} }(z) = \omega(z) \sum_{k=0}^{N-1} \frac{1}{h_k} p_k(z) p_k(\overline{z}),
    \\
    & R_{N,1}^{ \mathbb{H} }(z) = \omega(z) (\overline{z}-z) \sum_{k=0}^{N-1} \frac{1}{r_k} \Big( q_{2k+1}(z) q_{2k}(\overline{z}) - q_{2k}(z) q_{2k+1}(\overline{z}) \Big).
\end{split}
\end{equation}
These formulas serve as the fundamental building blocks for establishing our main theorems.

\section{Proofs} \label{Section_Proofs}

In this section, we prove our main results, Theorems~\ref{thm: main}, \ref{Thm_elliptic and MP}, \ref{Thm_eGinibre 2}, and \ref{Thm_eGinUSE genus expansion}. 

\subsection{Proof of Theorem~\ref{thm: main}}
 
By using \eqref{moments r1} and \eqref{1-point}, we have 
\begin{equation}
\begin{split} \label{M p1p2 N in terms of OP}
    & M^\mathbb{C}_{p_1,p_2,N} = \int_\mathbb{C} z^{p_1} \overline{z}^{p_2} \sum_{k=0}^{N-1} \frac{1}{h_k} p_k(z) p_k(\overline{z}) \,\omega(z) \, dA(z),
    \\
    & M^\mathbb{H}_{p_1,p_2,N} = \int_\mathbb{C} z^{p_1} \overline{z}^{p_2} \sum_{k=0}^{N-1} \frac{1}{r_k} (\overline{z}-z) \Big( q_{2k+1}(z) q_{2k}(\overline{z}) - q_{2k}(z) q_{2k+1}(\overline{z}) \Big) \, \omega(z) \, dA(z).
\end{split}
\end{equation}

We begin by considering the moments of complex ensembles and proving \eqref{eq: main1}. The argument is straightforward: applying the orthogonality relation \eqref{OP norm} together with the definition \eqref{eq: Tp}, we obtain
\begin{equation*}
\begin{split}
    M^\mathbb{C}_{p_1,p_2,N} & = \sum_{k=0}^{N-1} \frac{1}{h_k} \Big\langle \sum_{j_1=k-p_1}^{k+p_1} (A^{p_1})^{j_1}_k p_{j_1}, \sum_{j_2=k-p_2}^{k+p_2} (A^{p_2})^{j_2}_k p_{j_2} \Big\rangle
    = \sum_{k=0}^{N-1} \sum_{n=k- p_1 \wedge p_2}^{k+ p_1 \wedge p_2} \frac{h_n}{h_k} (A^{p_1})^n_k (A^{p_2})^n_k.
\end{split}
\end{equation*}
This leads to the desired identity \eqref{eq: main1}. The particular case $p_2=0$ follows directly.

Next, we show \eqref{eq: main2}. For this purpose, we first consider the representation of the linear map $T_p$ in terms of the skew-orthogonal basis $(q_k)_{k=1}^\infty$ as in \eqref{eq: Tp}: 
\begin{equation}
    T_p q_{k}(z) = \sum_{j=0}^{k+p} (B^p)^j_k q_j(z).
\end{equation}
We claim that the coefficients $(B^p)^j_k$ defined this way satisfy the formula given in \eqref{B in A}. Given the three-term recurrence relation \eqref{eq: recurrence} for the planar orthogonal polynomials, it was shown in \cite[Theorem 3.1]{AEP22} that
\begin{equation} \label{sop construction}
    q_{2k+1}(z) := p_{2k+1}(z), \qquad q_{2k}(z) := \sum_{j=0}^k \mu_{k,j} \ p_{2j}(z),
\end{equation}
form a family of skew orthogonal polynomials, where \( \mu_{k,j} \) is defined in \eqref{mu in sop}, and the skew norm is given by \eqref{skew norm}. The inverse transformation of \eqref{sop construction} takes the form 
\begin{equation} \label{sop inversion}
    p_{2k+1}(z) = q_{2k+1}(z), \qquad p_{2k}(z) = q_{2k}(z) - \lambda_{k-1} q_{2k-2}(z),
\end{equation}
where $q_{-2}(z)\equiv q_{-1}(z)\equiv0$, and $\lambda_k$ is given in \eqref{mu in sop}. 
Using \eqref{sop construction} and \eqref{sop inversion} as change-of-basis relations, we obtain \eqref{B in A}.

Then, it follows from \eqref{M p1p2 N in terms of OP} that 
\begin{align}  \label{eq: MH in B}
\begin{split}
 &\quad  M^\mathbb{H}_{p_1,p_2,N}
    = \sum_{k=0}^{N-1} \frac{1}{r_k} \bigg[ \Big\langle \sum_{j_1=0}^{2k+1+p_1} (B^{p_1})^{j_1}_{2k+1} q_{j_1}, \sum_{j_2=0}^{2k+p_2} (B^{p_2})^{j_2}_{2k} q_{j_2} \Big\rangle_s  - \Big\langle \sum_{j_1=0}^{2k+p_1} (B^{p_1})^{j_1}_{2k} q_{j_1}, \sum_{j_2=0}^{2k+1+p_2} (B^{p_2})^{j_2}_{2k+1} q_{j_2} \Big\rangle_s \bigg]
    \\
    & = \frac{1}{2} \sum_{k=0}^{N-1} \sum_{n=0}^{k+\frac{p_1 \vee p_2+1}{2}} \frac{r_n}{r_k} \Big( (B^{p_1})^{2n+1}_{2k+1} (B^{p_2})^{2n}_{2k} - (B^{p_1})^{2n}_{2k+1} (B^{p_2})^{2n+1}_{2k} -  (B^{p_1})^{2n+1}_{2k} (B^{p_2})^{2n}_{2k+1} +_ (B^{p_1})^{2n}_{2k} (B^{p_2})^{2n+1}_{2k+1} \Big). 
\end{split}
\end{align} 
Here, we have used 
\begin{equation*}
    \int_\mathbb{C} (\overline{z}-z) f(z) g(\overline{z}) \, \omega(z) \, dA(z) = \frac{1}{2} \langle f,g \rangle_s
\end{equation*}
which holds for polynomials with real coefficients since $\omega(z) \, dA(z)$ have real moments. Moreover, since \( (A^p)^j_k \) is non-zero only when \( |j - k| \le p \), the range of the index \( n \) in \eqref{eq: MH in B} can be restricted as in \eqref{def of mathfrak m jk}.


Finally, we show \eqref{eq: holomorpic moments}. 
In order to consider the case $(p_1,p_2)=(p,0)$, notice that $(B_0)^j_k=\delta_{j,k}$ as in \eqref{Ap trivial}. Then we have 
\begin{equation*}
\begin{split}
    M^{\mathbb{H}}_{p,0,N} & = \frac{1}{2} \sum_{k=0}^{N-1} (B^p)^{2k+1}_{2k+1}+(B^p)^{2k}_{2k}
   = \frac{1}{2} \sum_{k=0}^{N-1} \Big( (A^p)^{2k+1}_{2k+1}+ \sum_{j=0}^{k}\mu_{k,j}(A^p)^{2k}_{2j} - \lambda_k\sum_{j=0}^{k}\mu_{k,j}(A^p)^{2k+2}_{2j} \Big)
    \\
    & = \frac{1}{2} \sum_{k=0}^{2N-1} (A^p)^k_k + \frac{1}{2}\sum_{k=0}^{N-1} \Big( \sum_{j=0}^{k-1}\mu_{k,j}(A^p)^{2k}_{2j} - \sum_{j=0}^{k}\mu_{k+1,j}(A^p)^{2k+2}_{2j} \Big).
\end{split}
\end{equation*}
The first term in the final expression matches the form of \eqref{holomorphic moments}, while the second term can be simplified as  
\begin{equation*}
\begin{split}
     \sum_{k=0}^{N-1} \Big( \sum_{j=0}^{k-1}\mu_{k,j}(A^p)^{2k}_{2j} - \sum_{j=0}^{k}\mu_{k+1,j}(A^p)^{2k+2}_{2j} \Big)
    &=  \Big(\sum_{k=1}^{N-1} \sum_{j=0}^{k-1} \mu_{k,j}(A^p)^{2k}_{2j} - \sum_{k=1}^{N} \sum_{j=0}^{k-1} \mu_{k,j}(A^p)^{2k}_{2j}\Big)
     = - \sum_{j=0}^{N-1}\mu_{N,j}(A^p)^{2N}_{2j}.
\end{split}
\end{equation*}
Hence we obtain \eqref{eq: holomorpic moments}. This completes the proof. \qed

\subsection{Proof of Theorem~\ref{Thm_elliptic and MP}}
First, we observe the asymptotic behaviour of $(A^p)^j_k$ in Proposition~\ref{Prop_Ap jk explicit}.
\begin{lem} \label{lem: asymptotics}
Suppose $|r| \le p$ is a fixed integer.
\begin{itemize} 
    \item For the planar Hermite polynomial \eqref{def of planar Hermite}, we have for large $k$,
    \begin{equation} \label{Ap asymptotics}
        (A^p)^{k-r}_k = \tau^{\frac{p+r}{2}} \binom{p}{\frac{p+r}{2}} k^{\frac{p+r}{2}} + O(k^{\frac{p+r}{2}-1}).
    \end{equation}
    \item  For the planar Laguerre polynomial \eqref{def of planar Laguerre}, we have for large $k$ and $\nu$,
    \begin{equation} \label{Ap asymp laguerre}
        (A^p)^{k-r}_{k} = \sum_{l=0}^p \sum_{s\ge0} \frac{p!}{(p-l)!} \frac{2^{l-r-2s}}{s!(s+r)!(l-r-2s)!} \nu^{p-l} k^{l-s} (k+\nu)^{s+r} + O\big( (k+\nu)^{p+r-1} \big).
    \end{equation}
\end{itemize}
\end{lem}

\begin{proof}
The first identity in \eqref{Ap asymptotics} follows directly from \eqref{Ap Hermite}. For the second identity \eqref{Ap asymp laguerre}, by applying \( s \mapsto s + k \) in \eqref{Ap laguerre}, we obtain
\begin{equation*}
\begin{split}
    (A^p)^{k-r}_k & = \sum_{l=0}^p \frac{p!}{(p-l)!} \frac{\Gamma(p+\nu+1)}{\Gamma(l+\nu+1)} \sum_{s\ge0} \frac{k!}{(k+s-l)!} \frac{\Gamma(k+s+\nu+1)}{\Gamma(k-r+\nu+1)} \frac{2^{l-r-2s}}{s!(s+r)!(l-r-2s)!}
    \\
    & = \sum_{l=0}^p \frac{p!}{(p-l)!} \nu^{p-l} \sum_{s\ge0} k^{l-s} (k+\nu)^{s+r} \frac{2^{l-r-2s}}{s!(s+r)!(l-r-2s)!} + O( (k+\nu)^{p+r-1} ),
\end{split}
\end{equation*}
which gives \eqref{Ap asymp laguerre}.
\end{proof}

\begin{proof}[Proof of Theorem~\ref{Thm_elliptic and MP} (i)]
We first show \eqref{elliptic leading order} for the complex case. It follows from Theorem \ref{thm: main} and \eqref{Ap Hermite} that 
\begin{equation}
    M^{ \rm H,\mathbb{C}}_{p_1,p_2,N} = \sum_{k=0}^{N-1} \sum_{ r \in \mathcal{I}_{p_1 \wedge p_2} } \frac{(k-r)!}{k!} (A^{p_1})^{k-r}_k (A^{p_2})^{k-r}_k
\end{equation}
with $\mathcal{I}_p$ as in \eqref{I_p def}.
Applying \eqref{Ap asymptotics}, we have
\begin{equation}
\begin{split}
    M^{ \rm H,\mathbb{C}}_{p_1,p_2,N} & = \sum_{k=0}^{N-1} k^{ \frac{p_1+p_2}{2} } \bigg[ \sum_{ r \in \mathcal{I}_{p_1 \wedge p_2} } \tau^{ \frac{p_1+p_2}{2} + r } \binom{p_1}{\frac{p_1+r}{2}} \binom{p_2}{\frac{p_2+r}{2}} \bigg] + O( k^{ \frac{p_1+p_2}{2} - 1 } )
    \\
    & = \frac{1}{ \frac{p_1+p_2}{2}+1 } N^{ \frac{p_1+p_2}{2} + 1 } \bigg[ \sum_{ r \in \mathcal{I}_{p_1 \wedge p_2} } \tau^{ \frac{p_1+p_2}{2} + r } \binom{p_1}{\frac{p_1+r}{2}} \binom{p_2}{\frac{p_2+r}{2}} \bigg] + O( N^{ \frac{p_1+p_2}{2} } ), 
\end{split}
\end{equation}
where the last equality follows from the fact that $\sum_{k=1}^N k^p = \frac{1}{p+1} N^{p+1} + O(N^p).$ This gives the desired asymptotic behaviour \eqref{elliptic leading order} for $M_{p_1,p_2,N}^{ { \rm H }, \mathbb{C} }$. 

The assertion for the symplectic case \( M_{p_1,p_2,N}^{ {\rm H}, \mathbb{H} } \) will be addressed in a later subsection, where we establish Theorem~\ref{Thm_eGinUSE genus expansion}.
\end{proof}

\begin{proof}[Proof of Theorem~\ref{Thm_elliptic and MP} (ii)] 
By Theorem \ref{thm: main} and \eqref{Ap laguerre}, we have
\begin{equation}
    M^{ \rm L, \mathbb{C} }_{p_1,p_2,N} = \sum_{k=0}^{N-1} \sum_{r=-p_1 \wedge p_2}^{p_1\wedge p_2} \frac{(k-r)!}{k!} \frac{\Gamma(k-r+\nu+1)}{\Gamma(k+\nu+1)} (A^{p_1})^{k-r}_{k} (A^{p_2})^{k-r}_{k}.
\end{equation}
Note that \eqref{Ap asymp laguerre} can also be written as
\begin{equation*}
    (A^p)^{k-r}_k = \tau^{p+r} \sum_{l=0}^p \sum_{s\ge0} 
    \binom{p}{l} 2^{l-r-2s} \binom{l}{s} \binom{l-s}{s+r} \nu^{p-l} k^{l-s} (k+\nu)^{s+r} + O\big( (k+\nu)^{p+r-1} \big).
\end{equation*}
Applying the binomial expansion, we obtain
\begin{equation*}
    [x^{l-r}] (kx^2+2kx+k+\nu)^l = \sum_{s\ge0} 2^{l-r-2s} \binom{l}{s} \binom{l-s}{r+s} k^{l-r-s} (k+\nu)^{s+r},
\end{equation*}
where $[x^p]f(x)$ denotes the coefficient of $x^p$ in the expansion of $f(x)$. Then, we have
\begin{equation} \label{double counting 1}
\begin{split}
    &\quad [x^{p-r}] \Big( k(x+1)^2 + \nu(x+1) \Big)^p = [x^{p-r}] \sum_{l=0}^p \binom{p}{l} (\nu x)^{p-l} (kx^2+2kx+k+\nu)^l
    \\
    & = \sum_{l=0}^p \binom{p}{l} \nu^{p-l} [x^{l-r}] (kx^2+2kx+k+\nu)^l
 = \sum_{l=0}^p \sum_{s\ge0} 
    \binom{p}{l} \nu^{p-l} 2^{l-r-2s} \binom{l}{s} \binom{l-s}{s+r} k^{l-r-s} (k+\nu)^{s+r}.
\end{split}
\end{equation}
On the other hand, we also have
\begin{equation} \label{double counting 2}
    [x^{p-r}] \Big( k(x+1)^2 + \nu(x+1) \Big)^p = [x^{p-r}] \sum_{l=0}^p \binom{p}{l} k^l \nu^{p-l} (x+1)^{p+l} = \sum_{l=0}^p \binom{p}{l} \binom{p+l}{p-r} k^l \nu^{p-l}.
\end{equation}
Combining \eqref{double counting 1} and \eqref{double counting 2}, we obtain 
\begin{equation*}
    k^{-r} (A^p)^{k-r}_k = \tau^{p+r} \sum_{l=0}^p \binom{p}{l} \binom{p+l}{p-r} k^l \nu^{p-l} + O\big( (k+\nu)^{p-1} \big).
\end{equation*}
Similarly, it follows from 
\begin{equation*}
    [x^{l+r}] (kx^2+2kx+k+\nu)^l = \sum_{s\ge0} 2^{l-r-2s} \binom{l}{s} \binom{l-s}{r+s} k^{l-s} (k+\nu)^{s}
\end{equation*}
that 
\begin{equation*}
\begin{split}
    (k+\nu)^{-r} (A^p)^{k-r}_k 
    & = \tau^{p+r} \sum_{l=0}^p \binom{p}{l} \binom{p+l}{p+r} k^l \nu^{p-l} + O\big((k+\nu)^{p-1}\big).
\end{split}
\end{equation*}
Combining the above, if $\nu=\alpha N + o(N),$ we have 
\begin{equation*}
\begin{split}
    M^{ \rm L, \mathbb{C} }_{p_1,p_2,N}
    & = \sum_{k=0}^{N-1} \sum_{|r| \le p_1 \wedge p_2} k^{-r} (k+\nu)^{-r} (A^{p_1})^{k-r}_{k} (A^{p_2})^{k-r}_{k} + O\big( (k+\nu)^{p_1+p_2-1} \big)
    \\
    & = \sum_{k=0}^{N-1} \sum_{|r| \le p_1 \wedge p_2} \sum_{l_1=0}^{p_1} \sum_{l_2=0}^{p_2} k^{l_1+l_2} \nu^{p_1+p_2-l_1-l_2} \binom{p_1}{l_1} \binom{p_1+l_1}{p_1-r} \binom{p_2}{l_2} \binom{p_2+l_2}{p_2+r} + O\big( (k+\nu)^{p_1+p_2-1} \big)
    \\
    & = N^{p_1+p_2+1} \sum_{|r| \le p_1 \wedge p_2} \sum_{l_1=0}^{p_1} \sum_{l_2=0}^{p_2} \frac{\alpha^{p_1+p_2-l_1-l_2}}{l_1+l_2+1} \binom{p_1}{l_1} \binom{p_1+l_1}{p_1-r} \binom{p_2}{l_2} \binom{p_2+l_2}{p_2+r} + O( N^{p_1+p_2} ),
\end{split}
\end{equation*}
which leads to \eqref{shifted elliptic leading order} for the complex case.

Next, we derive the asymptotic of $M^{\rm L,\mathbb{H}}_{p_1,p_2,N}$. For this, note that 
\begin{equation*}
    (A^p)^{k-r}_k = \mathfrak{a}_{p,r}+O\big((k+\nu)^{p+r-1}\big), \qquad   \mathfrak{a}_{p,r} := \tau^{p+r} \sum_{l=0}^p \binom{p}{l}\binom{p+l}{p+r} k^l \nu^{p-l} (k+\nu)^r = O\big((k+\nu)^{p+r}\big). 
\end{equation*}
We have shown
\begin{align*}
   & \mathfrak{b}_{p,r} := k^{-r}\mathfrak{a}_{p,r} = \tau^{p+r}\sum_{l=0}^p \binom{p}{l}\binom{p+l}{p-r} k^l \nu^{p-l} \\
    & \mathfrak{c}_{p,r} := (k+\nu)^{-r}\mathfrak{a}_{p,r} = \tau^{p+r}\sum_{l=0}^p \binom{p}{l}\binom{p+l}{p+r} k^l \nu^{p-l} ,
\end{align*}
and
\begin{equation*}
    \sum_{k=0}^{N-1} \sum_{r\in\mathbb{Z}} \mathfrak{b}_{p_1,r}\mathfrak{c}_{p_2,r} = L_1(p_1,p_2)N^{p_1+p_2+1} +O(N^{p_1+p_2}).
\end{equation*}

On the other hand, by \eqref{mu in sop} and \eqref{def of planar Laguerre}, we have 
\begin{equation*}
\begin{split}
     \frac{1}{2^2}\lambda_{k-r} = 2^2k(k+\nu)+O\big((k+\nu)^1\big), \quad \frac{1}{2^{2s}}\mu_{k,k-s} = k^s(k+\nu)^s+O\big((k+\nu)^{2s-1}\big).
\end{split}
\end{equation*}
Then, by \eqref{B in A}, we have that, for $n=k-r$ with $r\in\mathbb{Z}$,
\begin{align*}
    \frac{1}{2^{p+2r}}(B^p)^{2n+1}_{2k+1} &= \mathfrak{a}_{p,2r}+O\big((k+\nu)^{p+2r-1}\big), \\
    \frac{1}{2^{p+2r+1}}(B^p)^{2n}_{2k+1} &= \mathfrak{a}_{p,2r+1}-k(k+\nu)\mathfrak{a}_{p,2r-1}+O\big((k+\nu)^{p+2r}\big), \\
    \frac{1}{2^{p+2r-1}}(B^p)^{2n+1}_{2k} &= \sum_{s\ge0}k^s(k+\nu)^s\mathfrak{a}_{p,2r-2s-1}+O\big((k+\nu)^{p+2r-2}\big), \\
    \frac{1}{2^{p+2r}}(B^p)^{2n}_{2k} 
    & = \mathfrak{a}_{p,2r}+O\big((k+\nu)^{p+2r-1}\big).
\end{align*}

Recall the definition \eqref{def of mathfrak m jk} of $\mathfrak{m}_{p_1,p_2,k}$. We separately investigate the first and second line of \eqref{def of mathfrak m jk}. Note that by \eqref{def of planar Laguerre} and \eqref{laguerre recurrence}, the skew-norm \eqref{skew norm} for Laguerre polynomials is given by 
\begin{equation*}
    r_k = (1-\tau^2)^2 (2k+1)! \, \Gamma(2k+2\nu+2).
\end{equation*}
For the first line of \eqref{def of mathfrak m jk}, we have
\begin{equation*}
\begin{split}
    \frac{1}{2^{p_1+p_2+1}} \frac{r_n}{r_k} \Big((B^{p_1})^{2n+1}_{2k+1} (B^{p_2})^{2n}_{2k} + (B^{p_1})^{2n}_{2k} (B^{p_2})^{2n+1}_{2k+1}\Big)
    & = k^{-2r}(k+\nu)^{-2r} \mathfrak{a}_{p_1,2r}\mathfrak{a}_{p_2,2r} + O\big((k+\nu)^{p_1+p_2-1}\big),
\end{split}
\end{equation*}
which gives rise to 
\begin{equation*}
    \frac{1}{2^{p_1+p_2+1}} \sum_{n\in\mathbb{Z}} \frac{r_n}{r_k} \Big((B^{p_1})^{2n+1}_{2k+1} (B^{p_2})^{2n}_{2k} + (B^{p_1})^{2n}_{2k} (B^{p_2})^{2n+1}_{2k+1}\Big) = \sum_{r=\text{even}} \mathfrak{b}_{p_1,r} \mathfrak{c}_{p_2,r} + O\big((k+\nu)^{p_1+p_2-1}\big).
\end{equation*}
Next, observe that the leading order of the second line of \eqref{def of mathfrak m jk} simplifies as 
\begin{equation*}
\begin{split}
    & \quad \frac{1}{2^{p_1+p_2}} \sum_{n\in\mathbb{Z}} \frac{r_n}{r_k} (B^{p_1})^{2n}_{2k+1}(B^{p_2})^{2n+1}_{2k}
    \\
    & = \sum_{r\in\mathbb{Z}} \sum_{s\ge0} k^{s-2r}(k+\nu)^{s-2r} \mathfrak{a}_{p_1,2r+1}\mathfrak{a}_{p_2,2r-2s-1}
    \\
    & \quad - \sum_{r\in\mathbb{Z}} \sum_{s\ge0} k^{s-2r+1}(k+\nu)^{s-2r+1} \mathfrak{a}_{p_1,2r-1}\mathfrak{a}_{p_2,2r-2s-1} + O\big((k+\nu)^{p_1+p_2-1}\big)
    \\
    & =  \sum_{r\in\mathbb{Z}} \Big( \sum_{s\ge0} k^{s-2r}(k+\nu)^{s-2r} \mathfrak{a}_{p_1,2r+1}\mathfrak{a}_{p_2,2r-2s-1} - \sum_{s\ge1} k^{s-1-2r}(k+\nu)^{s-1-2r} \mathfrak{a}_{p_1,2r+1}\mathfrak{a}_{p_2,2r-2(s-1)-1} \Big)
    \\
    & \quad -k^{-2r-1}(k+\nu)^{-2r-1} \mathfrak{a}_{p_1,2r+1}\mathfrak{a}_{p_2,2r+1}  + O\big((k+\nu)^{p_1+p_2-1}\big)
    \\
    & = -k^{-2r-1}(k+\nu)^{-2r-1} \mathfrak{a}_{p_1,2r+1} \mathfrak{a}_{p_2,2r+1} + O\big((k+\nu)^{p_1+p_2-1}\big).
\end{split}
\end{equation*}
Here, the second identity follows from shifting the indices $r\mapsto r+1$, $s\mapsto s+1$.
Therefore, we have
\begin{equation*}
    \frac{1}{2^{p_1+p_2+1}} \sum_{n\in\mathbb{Z}} \frac{r_n}{r_k} \Big( (B^{p_1})^{2n}_{2k+1}(B^{p_2})^{2n+1}_{2k}+(B^{p_1})^{2n}_{2k+1}(B^{p_2})^{2n+1}_{2k}\Big) = \sum_{r=\text{odd}} \mathfrak{b}_{p_1,r} \mathfrak{c}_{p_2,r} + O\big((k+\nu)^{p_1+p_2-1}\big).
\end{equation*} 
Thus we conclude that if $\nu=\alpha N+o(N)$,
\begin{equation*}
\begin{split}
    M^{\rm L, \mathbb{H}}_{p_1,p_2,N} = \sum_{k=0}^{N-1} \sum_{r\in\mathbb{Z}} \mathfrak{b}_{p_1,r} \mathfrak{c}_{p_2,r} + O\big((k+\nu)^{p_1+p_2-1}\big) = L_1(p_1,p_2)N^{p_1+p_2+1} +O(N^{p_1+p_2}),
\end{split}
\end{equation*}
which completes the proof. 
\end{proof}

\subsection{Proof of Theorem~\ref{Thm_eGinibre 2}} 

The proof of Theorem~\ref{Thm_eGinibre 2} relies on the application of suitable differential operators to the kernels in \eqref{def of kernel}, which serves to reduce the number of terms in the expression. This is followed by integration by parts, which facilitates the computation of the spectral moments.

\medskip

We first show \eqref{Thm 2.5_1}.  
In \cite[Proposition 2.3]{LR16}, it was shown that the kernel $K^{\rm H,\mathbb{C}}_{N}$ satisfies 
\begin{equation}
    \partial_z \Big( K^{\rm H,\mathbb{C}}_{N}(z,\overline{z}) \Big) = \frac{\omega^{\rm H}(z)}{1-\tau^2} \frac{1}{h_{N-1}} \Big( \tau \, p_N(z) p_{N-1}(\overline{z}) - p_{N-1}(z) p_N(\overline{z}) \Big).
\end{equation}
Integration by parts gives us that
\begin{equation*}
\begin{split}
    M^{\rm H,\mathbb{C}}_{p_1,p_2,N} 
    & = \frac{1}{1-\tau^2} \frac{1}{h_{N-1}} \int_\mathbb{C} \frac{1}{p_1+1} z^{p_1+1} \overline{z}^{p_2} \Big( p_{N-1}(z) p_N(\overline{z}) - \tau \, p_N(z) p_{N-1}(\overline{z}) \Big) \omega^{\rm H}(z) \, dA(z)
    \\
    & = \frac{1}{1-\tau^2} \frac{1}{p_1+1} \frac{1}{h_{N-1}} \bigg[ \Big\langle \sum_{j_1=N-p_1-2}^{N+p_1} (A^{p_1})^{j_1}_{N-1} \, p_{j_1}, \sum_{j_2=N-p_2}^{N+p_2} (A^{p_2})^{j_2}_N \, p_{j_2} \Big\rangle
    \\
    & \hspace{110pt} - \tau \Big\langle \sum_{j_1=N-p_1-1}^{N+p_1+1} (A^{p_1})^{j_1}_{N} \, p_{j_1}, \sum_{j_2=N-1-p_2}^{N-1+p_2} (A^{p_2})^{j_2}_{N-1} \, p_{j_2} \Big\rangle \bigg]
    \\
    & = \frac{1}{1-\tau^2} \frac{1}{p_1+1} \sum_{n=N-1-(p_1+1)\vee p_2}^{N+(p_1+1)\vee p_2} \frac{n!}{(N-1)!} \Big[ (A^{p_1+1})^n_{N-1} (A^{p_2})^n_N - \tau (A^{p_1+1})^n_N (A^{p_2})^n_{N-1} \Big].
\end{split}
\end{equation*} 
Therefore we obtain \eqref{Thm 2.5_1}.  

\medskip
Next, we show  \eqref{Thm 2.5_2}. By \cite[Proposition 1.1]{BE23}, we have 
\begin{equation}
    \Big(\partial_z-\frac{z}{1+\tau}\Big) \widehat{K}^{\rm H,\mathbb{H}}_{N}(z,w) = \frac{1}{2(1-\tau^2)} \widehat{K}^{\rm H,\mathbb{C}}_{2N}(z,w) - \frac{1}{2(1-\tau^2)} S_{N}(z,w)
\end{equation}
where
\begin{equation*}
    S_{N}(z,w) = \sum_{j=0}^{N-1} \frac{1}{h_{2N}} \frac{(2N)!!}{(2j)!!} p_{2N}(z) p_{2j}(w).
\end{equation*}

Note that
\begin{equation*}
    \Big(\partial_z+\frac{z}{1+\tau}\Big) \omega^{\rm H}(z) = \frac{1}{1-\tau^2} (z-\overline{z}) \omega^{\rm H}(z).
\end{equation*}
Then, integration by parts gives 
\begin{equation} \label{eq: partial z}
    M^{\rm H,\mathbb{H}}_{p_1,p_2,N} - (1-\tau^2) \int_\mathbb{C} p_1 z^{p_1-1} \overline{z}^{p_2} K^{\rm H,\mathbb{H}}_{N}(z,\overline{z}) \, dA(z) = \frac{1}{2} M^{\rm H,\mathbb{C}}_{p_1,p_2,2N} - \frac{1}{2} \int_\mathbb{C} z^{p_1} \overline{z}^{p_2} S_{N}(z,\overline{z}) \omega^{\rm H}(z) \, dA(z).
\end{equation}
A similar argument, with the roles of \( z \) and \( \overline{z} \) interchanged, leads to
\begin{equation} \label{eq: partial w}
    M^{\rm H,\mathbb{H}}_{p_1,p_2,N} + (1-\tau^2) \int_\mathbb{C} p_2 z^{p_1} \overline{z}^{p_2-1} K^{\rm H,\mathbb{H}}_{N}(z,\overline{z}) \, dA(z) = \frac{1}{2} M^{\rm H,\mathbb{C}}_{p_1,p_2,2N} - \frac{1}{2} \int_\mathbb{C} z^{p_1} \overline{z}^{p_2} S_{N}(\overline{z},z) \omega^{\rm H}(z) \, dA(z).
\end{equation}
Combining (\ref{eq: partial z}) and (\ref{eq: partial w}), we obtain
\begin{equation*}
\begin{split}
    & \quad \Big( \frac{1}{p_1} + \frac{1}{p_2} \Big) M^{\rm H,\mathbb{H}}_{p_1,p_2,N} - (1-\tau^2) M^{\rm H,\mathbb{H}}_{p_1-1,p_2-1,N} 
    \\
    & = \frac{1}{2} \Big( \frac{1}{p_1} + \frac{1}{p_2} \Big) M^{\rm H,\mathbb{C}}_{p_1,p_2,2N} - \frac{1}{2} \int_\mathbb{C} z^{p_1} \overline{z}^{p_2} \Big( \frac{1}{p_1} S_{N}(z,\overline{z}) + \frac{1}{p_2} S_{N}(\overline{z},z) \Big) \omega^{\rm H}(z) \, dA(z).
\end{split}
\end{equation*}
Note here that 
\begin{equation*}
\begin{split}
    \int_\mathbb{C} z^{p_1} \overline{z}^{p_2} \ S_{N}(z,\overline{z}) \omega^{\rm H} \, dA(z)
    & = \sum_{j=0}^{N-1} \frac{1}{(2N)! \sqrt{1-\tau^2}} \frac{(2N)!!}{(2j)!!} \sum_{m,n\ge0} (A^{p_1})^m_{2N} \ (A^{p_2})^n_{2j} \langle p_m, p_n \rangle
    \\
    & = \sum_{j=0}^{N-1} \sum_{n=2N-p_1}^{2N+p_1} \frac{n!}{(2N)!} \frac{(2N)!!}{(2j)!!} (A^{p_1})^n_{2N} \ (A^{p_2})^n_{2j},
\end{split}
\end{equation*} 
which completes the proof. \qed

\subsection{Proof of Theorem~\ref{Thm_eGinUSE genus expansion}}

In order to prove the Theorem~\ref{Thm_eGinUSE genus expansion} (i), we define 
\begin{equation} \label{def of m}
\begin{split}
    m_{p_1,p_2,N}^{l_1,l_2} :=  \sum_{r\equiv p}   \frac{\tau^{\frac{p_1+p_2}{2}+r}}{1-\tau^2} & \bigg[ \kappa^{(1)}_r \Big(N-\frac{p_1+r}{2}+l_1\Big)_{\frac{p_1-r}{2}-l_1+1} \Big(N-\frac{p_2+r}{2}+l_2+1\Big)_{\frac{p_2+r}{2}-l_2}
    \\
    & - \kappa^{(2)}_r \Big(N-\frac{p_1+r}{2}+l_1+1\Big)_{\frac{p_1-r}{2}-l_1+1} \Big(N-\frac{p_2+r}{2}+l_2+1\Big)_{\frac{p_2+r}{2}-l_2} \bigg],
\end{split}
\end{equation}
where
\begin{equation*}
\begin{split}
    & \kappa^{(1)}_r := \frac{1}{2^{l_1} l_1!} \frac{1}{2^{l_2} l_2!} \frac{p_1!}{(\frac{p_1-r}{2}-l_1+1)!(\frac{p_1+r}{2}-l_1)!} \frac{p_2!}{(\frac{p_2-r}{2}-l_2)!(\frac{p_2+r}{2}-l_2)!},
    \\
    & \kappa^{(2)}_r := \frac{1}{2^{l_1} l_1!} \frac{1}{2^{l_2} l_2!} \frac{p_1!}{(\frac{p_1-r}{2}-l_1+1)!(\frac{p_1+r}{2}-l_1)!} \frac{p_2!}{(\frac{p_2-r}{2}-l_2+1)!(\frac{p_2+r}{2}-l_2-1)!}. 
\end{split}
\end{equation*}
Here, \( r \equiv p \) is shorthand for the congruence \( r \equiv p_1 \equiv p_2 \pmod{2} \). Note that the summation in \eqref{def of m} is finite since we regard the reciprocal factorial \( \frac{1}{n!} \) as \( \frac{1}{\Gamma(n+1)} = 0 \) for negative integers \( n \), as explained in section \ref{Section_Preliminaries} after \eqref{laguerre linearisation}.
Then by Theorem~\ref{Thm_eGinibre 2} (i),
\begin{equation}
    M^{\rm H,\mathbb{C}}_{p_1,p_2,N} = \sum_{l_1=0}^{\lfloor p_1/2 \rfloor} \sum_{l_2=0}^{\lfloor p_2/2 \rfloor} m_{p_1,p_2,N}^{l_1,l_2}.
\end{equation}

Since $m_{p_1,p_2,N}^{l_1,l_2}$ is a polynomial in $N$, we consider its expansion
\begin{equation}
    m_{p_1,p_2,N}^{l_1,l_2} =: \sum_{g=0}^{\frac{p_1+p_2}{2}-l_1-l_2} \mathcal{C}_g(p_1,p_2,l_1,l_2) N^{\frac{p_1+p_2}{2}-l_1-l_2+1-g}.
\end{equation}
Then, one can observe that 
\begin{equation*}
    M^{\rm H,\mathbb{C}}_{p_1,p_2,N} = C_1(p_1,p_2) N^{\frac{p_1+p_2}{2}+1} + C_2(p_1,p_2) N^{\frac{p_1+p_2}{2}} + O(N^{\frac{p_1+p_2}{2}-1}),
\end{equation*}
and that
\begin{align}
 \label{eq: C1}
    C_1(p_1,p_2)& = \mathcal{C}_0(p_1,p_2,0,0),
    \\
     \label{eq: C2}
    C_2(p_1,p_2)& = \mathcal{C}_1(p_1,p_2,0,0) + \mathcal{C}_0(p_1,p_2,1,0) + \mathcal{C}_0(p_1,p_2,0,1).
\end{align} 

First, we check that \eqref{eq: C1} yields the same result as in \eqref{def of C1 p1 p2}. Indeed,
\begin{equation*}
\begin{split}
    \mathcal{C}_0(p_1,p_2,0,0) & = \frac{1}{1-\tau^2} \sum_{r\equiv p} \tau^{\frac{p_1+p_2}{2}+r} (\kappa^{(1)}_r-\kappa^{(2)}_r) \Big\vert_{l_1=l_2=0}
    \\
    & = \frac{1}{1-\tau^2} \sum_{r\equiv p} \tau^{\frac{p_1+p_2}{2}+r} \bigg[ \frac{1}{\frac{p_1-r}{2}+1} \binom{p_1}{\frac{p_1+r}{2}} \binom{p_2}{\frac{p_2+r}{2}} - \frac{1}{\frac{p_1+r}{2}} \binom{p_1}{\frac{p_1+r}{2}-1} \binom{p_2}{\frac{p_2+r}{2}-1} \bigg].
\end{split}
\end{equation*}
Then we have
\begin{equation*}
\begin{split}
    &\quad \Big(\frac{p_1+p_2}{2}+1\Big) \mathcal{C}_0(p_1,p_2,0,0)
    \\
    & = \frac{1}{1-\tau^2} \sum_{r\equiv p} \tau^{\frac{p_1+p_2}{2}+r} \bigg[ \bigg(1+\frac{\frac{p_2+r}{2}}{\frac{p_1-r}{2}+1}\bigg)\binom{p_1}{\frac{p_1+r}{2}} \binom{p_2}{\frac{p_2+r}{2}} - \bigg(1+\frac{\frac{p_2-r}{2}+1}{\frac{p_1+r}{2}}\bigg)\binom{p_1}{\frac{p_1+r-2}{2}} \binom{p_2}{\frac{p_2+r-2}{2}} \bigg]
    \\
    & = \frac{1}{1-\tau^2} \sum_{r\equiv p} \tau^{\frac{p_1+p_2}{2}+r} \bigg[ \binom{p_1}{\frac{p_1+r}{2}} \binom{p_2}{\frac{p_2+r}{2}} - \binom{p_1}{\frac{p_1+r-2}{2}} \binom{p_2}{\frac{p_2+r-2}{2}} \bigg]
 = \sum_{r\equiv p} \tau^{\frac{p_1+p_2}{2}+r} \binom{p_1}{\frac{p_1+r}{2}} \binom{p_2}{\frac{p_2+r}{2}}.
\end{split}
\end{equation*}

Next, we compute \eqref{eq: C2}. Using
\begin{equation*}
    (x-m)_n = x^n - \frac{n(2m-n+1)}{2} x^{n-1} + O(x^{n-2}),
\end{equation*}
we have
\begin{align*}
   \mathcal{C}_1(p_1,p_2,0,0) = &   \sum_{r\equiv p} \frac{ \tau^{\frac{p_1+p_2}{2}+r} }{1-\tau^2} \bigg[ -\frac{1}{2} \frac{1}{\frac{p_1-r}{2}+1}\binom{p_1}{\frac{p_1+r}{2}}\binom{p_2}{\frac{p_2+r}{2}} \bigg( \Big(\frac{p_1-r}{2}+1\Big)\frac{p_1+3r}{2} + \Big( \frac{p_2+r}{2}-1 \Big) \frac{p_2+r}{2} \bigg)
    \\
    & +\frac{1}{2} \frac{1}{\frac{p_1-r}{2}+1}\binom{p_1}{\frac{p_1+r}{2}}\binom{p_2}{\frac{p_2+r}{2}-1} \bigg( \Big(\frac{p_1-r}{2}+1\Big)\Big(\frac{p_1+3r}{2}-2\Big) + \Big( \frac{p_2+r}{2}-1 \Big) \frac{p_2+r}{2} \bigg) \bigg].
\end{align*} 
On the other hand, we also have 
\begin{align*}
 \mathcal{C}_0(p_1,p_2,1,0) & =  \sum_{r\equiv p}  \frac{ \tau^{\frac{p_1+p_2}{2}+r}   }{2(1-\tau^2) } \bigg[ \frac{p_1! \, p_2! }{(\frac{p_1-r}{2})!(\frac{p_1+r-2}{2})!(\frac{p_2-r}{2})!(\frac{p_2+r}{2})!} - \frac{p_1!\,p_2!}{(\frac{p_1-r}{2})!(\frac{p_1+r-2}{2})! (\frac{p_2-r+2}{2})!(\frac{p_2+r-2}{2})!} \bigg],
 \\
   \mathcal{C}_0(p_1,p_2,0,1) & =   \sum_{r\equiv p}  \frac{ \tau^{\frac{p_1+p_2}{2}+r} }{2(1-\tau^2) } \bigg[ \frac{p_1! \,p_2 !}{(\frac{p_1-r+2}{2})!(\frac{p_1+r}{2})!(\frac{p_2-r-2}{2})!(\frac{p_2+r-2}{2})!} - \frac{p_1!\,p_2!}{(\frac{p_1-r+2}{2})!(\frac{p_1+r}{2})! (\frac{p_2-r}{2})!(\frac{p_2+r-4}{2})!} \bigg].
\end{align*} 
Combining all of the above, straightforward computations give rise to 
\begin{equation*}
\begin{split}
    C_2(p_1,p_2)
    & = \frac{1}{1-\tau^2} \sum_{r\equiv p} \tau^{\frac{p_1+p_2}{2}+r} \bigg[ \binom{p_1}{\frac{p_1+r}{2}} \binom{p_2}{\frac{p_2+r}{2}} \Big( -\frac{r}{2} \Big) - \binom{p_1}{\frac{p_1+r-2}{2}} \binom{p_2}{\frac{p_2+r-2}{2}} \Big( -\frac{r-2}{2} \Big) \bigg]
    \\
    & = -\sum_{r\in\mathcal{I}_{p_1 \wedge p_2}} \tau^{\frac{p_1+p_2}{2}+r} \binom{p_1}{\frac{p_1+r}{2}} \binom{p_2}{\frac{p_2+r}{2}} \frac{r}{2}.
\end{split}
\end{equation*}
This shows Theorem~\ref{Thm_eGinUSE genus expansion} (i). 

\medskip Next, we show the second assertion (ii) of the theorem. 
Let
\begin{equation*}
    F(p_1,p_2,N) := \sum_{k=0}^{N-1} \sum_{n=2N - p_1 \vee p_2}^{2N + p_1 \vee p_2} \frac{n!}{(2N)!} \frac{(2N)!!}{(2k)!!} \bigg[ \frac{p_1}{p_1+p_2} (A^{p_1})^{n}_{2k} (A^{p_2})^{n}_{2N} + \frac{p_2}{p_1+p_2} (A^{p_1})^{n}_{2N} (A^{p_2})^{n}_{2k} \bigg].
\end{equation*}
Putting $r=2N-n$, $s=N-k$ and considering that $(A^p)^\alpha_\beta\neq0$ only if $|\alpha-\beta|\le p$, we have 
\begin{equation*}
\begin{split}
    F(p_1,p_2,N) & = \sum_{r\in\mathcal{I}_{p_1 \vee p_2}} \sum_{s=1}^{(p_1+p_2)/2} \frac{(2N-r)!}{(2N)!} \frac{(2N)!!}{(2N-2s)!!} 
    \\
    & \qquad \times \bigg[  \frac{p_1}{p_1+p_2} (A^{p_1})^{2N-r}_{2N-2s} (A^{p_2})^{2N-r}_{2N} + \frac{p_2}{p_1+p_2} (A^{p_1})^{2N-r}_{2N} (A^{p_2})^{2N-r}_{2N-2s} \bigg].
\end{split}
\end{equation*}
Then, by \eqref{Ap asymptotics}, we obtain 
\begin{equation*}
\begin{split}
    F(p_1,p_2,N) & = (2N)^{\frac{p_1+p_2}{2}} \sum_{r\in\mathcal{I}_{p_1 \vee p_2}} \sum_{s=1}^{(p_1+p_2)/2} \tau^{\frac{p_1+p_2}{2}+r-s}
    \\
    & \qquad \times \bigg[ \frac{p_1}{p_1+p_2} \binom{p_1}{\frac{p_1+r}{2}-s} \binom{p_2}{\frac{p_2+r}{2}} + \frac{p_2}{p_1+p_2} \binom{p_1}{\frac{p_1+r}{2}} \binom{p_2}{\frac{p_2+r}{2}-s} \bigg] + O( N^{\frac{p_1+p_2}{2}-1}).
\end{split}
\end{equation*}
Then, by Theorem~\ref{Thm_eGinibre 2}~(ii), the desired behaviour \eqref{eq: C2prime} follows by induction. This completes the proof. \qed

\appendix

\section{Moments of the elliptic law} \label{Appendix_moment elliptic}
Here, we give a direct computation of \eqref{eq_elliptic law moment} using the conformal map.

\begin{prop} \label{Lem_mixed moments elliptic law}
Suppose that $p_1+p_2$ is even.  
Then we have 
\begin{equation} \label{elliptic integral mixed}
\frac{1}{1-\tau^2}\int_S z^{p_1} \overline{z}^{p_2} \,dA(z) 
=  C_1(p_1,p_2),
\end{equation}
where $S$ and $C_1(p_1,p_2)$ are given by \eqref{def of ellipitc law} and \eqref{def of C1 p1 p2}, respectively.  
\end{prop}
\begin{proof}
Let us denote by
\begin{equation}
f(z)= z+\frac{\tau}{z}, \qquad f: \mathbb{D}^c \to S^c,   
\end{equation}
the Joukowsky transform.
By applying Green’s formula, and a change of variables we obtain
\begin{align*}
\int_S z^{p_1} \overline{z}^{p_2} \,dA(z) &= \frac{1}{p_2+1}\frac{1}{2\pi i}  \int_{ \partial S } z^{p_1} \overline{z}^{ p_2+1 } \,dz = \frac{1}{p_2+1}\frac{1}{2\pi i}  \int_{ \partial \mathbb{D} } f(w)^{p_1} f(\overline{w})^{ p_2+1 } \,f'(w)\,dw
\\
&=  \frac{1}{p_2+1}\frac{1}{2\pi i}  \int_{ \partial \mathbb{D} } f(w)^{p_1} f(1/w)^{ p_2+1 } \,f'(w)\,dw
\\
&=  \frac{1}{p_2+1}\frac{1}{2\pi i}  \int_{ \partial \mathbb{D} } \frac{ (w^2+\tau)^{p_1} (\tau w^2+1)^{p_2+1} (w^2-\tau)  }{ w^{ p_1+p_2+3 } } \,dw. 
\end{align*}
Recall that $[z^n]f(z)$ denotes the coefficient of $z^n$ in $f(z)$. 
Then we have 
\begin{equation*}
\int_S z^{p_1} \overline{z}^{p_2} \,dA(z)= \frac{1}{p_2+1} [w^{ p_1+p_2+2 }] (w^2+\tau)^{p_1} (\tau w^2+1)^{p_2+1} (w^2-\tau).
\end{equation*}
Note that
\begin{align*}
&\quad [w^{p_1+p_2+2}] (w^2+\tau)^{p_1} (\tau w^2+1)^{p_2+1} (w^2-\tau)
\\
&= [w^{p_1+p_2}] (w^2+\tau)^{p_1} (\tau w^2+1)^{p_2+1} - \tau [w^{p_1+p_2+2}] (w^2+\tau)^{p_1} (\tau w^2+1)^{p_2+1}.
\end{align*}
Furthermore, we have
\begin{align*}
[w^{p_1+p_2}] (w^2+\tau)^{p_1} (\tau w^2+1)^{p_2+1}
&= \sum_{l=0}^{p_1} \binom{p_1}{l } \binom{p_2+1}{ \frac{p_2-p_1}{2}+l } \tau^{ (p_2-p_1)/2+2l },
\\
[w^{p_1+p_2+2}] (w^2+\tau)^{p_1} (\tau w^2+1)^{p_2+1}
&= \sum_{l=0}^{p_1} \binom{p_1}{l } \binom{p_2+1}{ \frac{p_2-p_1}{2}+l+1 } \tau^{ (p_2-p_1)/2+2l+1 }.
\end{align*}
Then, we have
\begin{equation*}
\begin{split}
    \int_S z^{p_1} \overline{z}^{p_2} \,dA(z)
    & = \frac{1}{p_2+1} \sum_{l=0}^{p_1} \bigg[ \tau^{ (p_2-p_1)/2+2l } \binom{p_1}{l} \binom{p_2+1}{\frac{p_2-p_1}{2}+l} - \tau^{ (p_2-p_1)/2+2l+2 } \binom{p_1}{l} \binom{p_2+1}{\frac{p_2-p_1}{2}+l+1} \bigg]
    \\
    & = \sum_{l=0}^{p_1+1} \tau^{ (p_2-p_1)/2+2l } \bigg[ \frac{1}{\frac{p_1+p_2}{2}-l+1} \binom{p_1}{l} \binom{p_2}{\frac{p_2-p_1}{2}+l} - \frac{1}{\frac{p_2-p_1}{2}+l}\binom{p_1}{l-1} \binom{p_2}{\frac{p_2-p_1}{2}+l-1} \bigg].
\end{split}
\end{equation*}
By rearranging the terms, it follows that
\begin{equation*}
    \frac{1}{1-\tau^2} \int_S z^{p_1} \overline{z}^{p_2} \,dA(z) = \frac{1}{\frac{p_1+p_2}{2}+1} \sum_{l=0}^{p_1} \tau^{ (p_2-p_1)/2+2l } \binom{p_1}{l} \binom{p_2}{\frac{p_2-p_1}{2}+l}.
\end{equation*}
Setting $r=2l-p_1$, we obtain \eqref{def of C1 p1 p2}.
\end{proof}

\section{Spectral moments of the elliptic Ginibre ensembles} \label{Appendix_eGinibre spectral moments}
We present an explicit formula for the spectral moments of the elliptic Ginibre ensembles. It is a direct consequence of \eqref{Ap Hermite} and Theorem \ref{thm: main}. Let $\mathcal{I}_p$ as in \eqref{I_p def}. 
We also define  for $p_1,p_2$ even:
\begin{equation}
\begin{split}
    f_{k,s,l_1,l_2}(p_1,p_2)
    & := \sum_{r=-\frac{p_1}{2}+l_1}^{\frac{p_1}{2}-l_1} \tau^{\frac{p_1+p_2}{2}+2r-s} \frac{(2k+1-2r)!}{(2k+1)!}\frac{ \binom{2k+1}{\frac{p_1}{2}-l_1+r}  }{(\frac{p_1}{2}-l_1-r)!}   \frac{ \binom{2k-2s}{\frac{p_2}{2}-l_2+r-s}  }{(\frac{p_2}{2}-l_2+s-r)!}  
    \\
    & \quad - \tau^{\frac{p_1+p_2}{2}+2r-s-1} \frac{(2k+2-2r)!}{(2k+1)!}\frac{ \binom{2k+1}{\frac{p_1}{2}-l_1+r}  }{(\frac{p_1}{2}-l_1-r)!}   \frac{  \binom{2k-2s}{\frac{p_2}{2}-l_2+r-s-1}  }{(\frac{p_2}{2}-l_2+s-r+1)!} ,
\end{split}
\end{equation}
and for $p_1,p_2$ odd:
\begin{equation}
\begin{split}
f_{k,s,l_1,l_2}(p_1,p_2) 
    & := \sum_{r=-\frac{p_1+1}{2}+l_1}^{\frac{p_1+1}{2}-l_1} - \tau^{\frac{p_1+p_2}{2}+2r-s} \frac{(2k+1-2r)!}{(2k+1)!}\frac{ \binom{2k+1}{\frac{p_1+1}{2}-l_1+r} }{(\frac{p_1-1}{2}-l_1-r)!}   \frac{ \binom{2k-2s}{\frac{p_2-1}{2}-l_2+r-s} }{(\frac{p_2+1}{2}-l_2+s-r)!}  
    \\
    & \quad + \tau^{\frac{p_1+p_2}{2}+2r-s-1} \frac{(2k+2-2r)!}{(2k+1)!}\frac{ \binom{2k+1}{\frac{p_1-1}{2}-l_1+r} }{(\frac{p_1+1}{2}-l_1-r)!}   \frac{ \binom{2k-2s}{\frac{p_2-1}{2}-l_2+r-s} }{(\frac{p_2+1}{2}-l_2+s-r)!}  .
\end{split}
\end{equation}

\begin{cor} \label{Cor_eGinibre moments}
    We have the following. 
    \begin{itemize}
        \item[\textup (i)]  For the complex elliptic Ginibre ensemble
    \begin{align}
    \begin{split}
        M_{p_1,p_2,N}^{\rm H, \mathbb{C}} &= 
         \sum_{k=0}^{N-1} \sum_{ r \in \mathcal{I}_{p_1 \wedge p_2} } \sum_{l_1=0}^{ \lfloor p_1/2 \rfloor } \sum_{l_2=0}^{ \lfloor p_2/2 \rfloor } \tau^{ \frac{p_1+p_2}{2} + r } \frac{(k-r)!}{k!} 
        \\
        &\quad \times \frac{p_1!}{2^{l_1} \ l_1! (\frac{p_1-r}{2}-l_1)!} \binom{k}{\frac{p_1+r}{2}-l_1} \frac{p_2!}{2^{l_2} \ l_2! (\frac{p_2-r}{2}-l_2)!} \binom{k}{\frac{p_2+r}{2}-l_2}.
    \end{split}
    \end{align}
    In particular,
    \begin{equation}
        M_{2p,0,N}^{\rm H, \mathbb{C}} = \tau^p (2p-1)!! \sum_{k=0}^{N-1} \sum_{l=0}^{p} 2^l \binom{p}{l} \binom{k}{l} =  \tau^p M_{2p,N}^{ \rm GUE }. 
    \end{equation}
    \item[\textup (ii)] For the symplectic elliptic Ginibre ensemble
\begin{equation}
\begin{split}
    & M_{p_1,p_2,N}^{\rm H, \mathbb{H} } = \frac{1}{2} \sum_{k=0}^{N-1} \sum_{s=0}^{k\wedge\frac{p_1+p_2}{2}} \sum_{l_1=0}^{p_1/2} \sum_{l_2=0}^{p_2/2} \frac{(2k)!!}{(2k-2s)!!} \frac{p_1!}{2^{l_1} l_1!} \frac{p_2!}{2^{l_2} l_2!} \bigg\{ f_{k,s,l_1,l_2}(p_1,p_2) + f_{k,s,l_1,l_2}(p_2,p_1) \bigg\}. 
\end{split}
\end{equation}
In particular,
\begin{align} \label{eq: moment eGinSE}
  M_{2p,0,N}^{\rm H, \mathbb{H} } = 
  & \frac{1}{2} M_{2p,0,2N}^{ H, \mathbb{C} }
  + \frac{1}{2} \sum_{r=1}^{p} \sum_{l=0}^p \tau^{p-r} \frac{(2N)!!}{(2N-2r)!!} \frac{(2p)!}{2^l l! (p-l+r)!} \binom{2N-2r}{p-l-r}.
\end{align}
    \end{itemize}
\end{cor}

\section{Verification of some combinatorial identities} \label{Appendix_combinatorial identities}
Here, we provide the proof of some identities given in previous sections.

To show \eqref{L1 p1p2 tau1 Narayana}, notice that 
\begin{equation*}
\begin{split}
    L_1(p_1,p_2)  \Big|_{ \tau=1 } 
    & = \sum_{r=-p_1 \wedge p_2}^{p_1\wedge p_2} \sum_{l_1=0}^{p_1} \sum_{l_2=0}^{p_2} \frac{\alpha^{p-l_1-l_2}}{l_1+l_2+1} \binom{p_1}{l_1} \binom{p_1+l_1}{l_1+r} \binom{p_2}{l_2} \binom{p_2+l_2}{l_2-r}
    \\
    & = \sum_{l_1=0}^{p_1} \sum_{l_2=0}^{p_2} \frac{\alpha^{p-l_1-l_2}}{l_1+l_2+1} \binom{p_1}{l_1} \binom{p_2}{l_2} \binom{p+l_1+l_2}{l_1+l_2}
    \\
    & = \sum_{l=0}^{p} \frac{\alpha^{p-l}}{l+1} \binom{p}{l} \binom{p+l}{l}  = \sum_{l=0}^{p} \frac{\alpha^l}{p+1} \binom{p+1}{l} \binom{2p-l}{p}.
\end{split}
\end{equation*}
This is an equivalent expression for the Narayana polynomials $N_p(1+\alpha)$, see e.g. \cite[Eq. (2.5)]{MS08}.

In order to prove \eqref{C2 prime tau1}, notice that 
\begin{equation*}
    C'_2(p_1,p_2) \Big\vert_{\tau\to1}
    = -\frac{1}{2} \sum_{r\equiv p} \sum_{s=1}^{p} \bigg[ \frac{p_1}{p_1+p_2} \binom{p_1}{\frac{p_1+r}{2}-s} \binom{p_2}{\frac{p_2-r}{2}} + \frac{p_2}{p_1+p_2} \binom{p_1}{\frac{p_1-r}{2}} \binom{p_2}{\frac{p_2+r}{2}-s} \bigg].
\end{equation*}
Setting $p=(p_1+p_2)/2$ and $t=(p_2-r)/2$, we have
\begin{equation*}
    \sum_{r\equiv p} \binom{p_1}{\frac{p_1+r}{2}-s} \binom{p_2}{\frac{p_2-r}{2}} = \sum_{t=0}^{p_2} \binom{p_1}{p-s-t} \binom{p_2}{t} = \binom{2p}{p-s}.
\end{equation*}
Likewise, it holds that
\begin{equation*}
    \sum_{r\equiv p} \binom{p_1}{\frac{p_1-r}{2}} \binom{p_2}{\frac{p_2+r}{2}-s} = \binom{2p}{p-s}.
\end{equation*}
Thus we have
\begin{equation*}
    C'_2(p_1,p_2) \Big\vert_{\tau\to1} = -\frac{1}{2} \sum_{s=1}^{p} \binom{2p}{p-s} = -\frac{1}{2} \sum_{l=0}^{p-1} \binom{2p}{l}.
\end{equation*}

\medskip

\subsection*{Data availability statement} There is no data associated to this work.

\subsection*{Conflict of interest statement} The authors have no conflicts of interest to disclose.

\bibliographystyle{abbrv}

\end{document}